\documentclass{article}

 \usepackage[vmargin=1.18in,hmargin=1.1in]{geometry}
 \usepackage{amsmath,amssymb,amsfonts,graphicx,amsthm,nicefrac,mathtools,bm}
 \usepackage{hyperref}
 \usepackage{natbib}
 \usepackage[english]{babel}
 \usepackage{times}
 \usepackage[T1]{fontenc}
 \usepackage{bbm,mdframed}
 \usepackage[dvipsnames]{xcolor}
 \usepackage{xspace}
 \usepackage{rotating,multirow}
 \usepackage{graphics,tikz}
 \usepackage{grffile}
 \usepackage{subcaption}
 \usepackage{epstopdf} % for postscript graphics files 
 \usepackage{listings}
\usepackage[boxruled,linesnumbered]{algorithm2e}
\usepackage{float}

 \newtheorem{theorem}{Theorem}

 \newtheorem{lemma}{Lemma}
  
 \newtheorem{corollary}{Corollary}

 \theoremstyle{definition}
 \newtheorem{definition}{Definition}
 \theoremstyle{remark}
 \newtheorem{remark}{Remark}
 
 \theoremstyle{example}
 
 \makeatletter
 
 \newtheorem*{rep@theorem}{\rep@title}
 \newcommand{\newreptheorem}[2]
 {\newenvironment{rep#1}[1]
 	{\def\rep@title{#2 \ref{##1}} \begin{rep@theorem}}%
 		{\end{rep@theorem}}}
 \makeatother
 \newreptheorem{theorem}{Theorem}
 \newreptheorem{lemma}{Lemma}
 \newreptheorem{corollary}{Corollary}
 \newreptheorem{proposition}{Proposition}

 \newcommand{\one}[1]{{\mathbbm{1}}_{{#1}}}
 \newcommand{\PP}[1]{\textnormal{Pr}\!\left\{{#1}\right\}} % Probability
 \newcommand{\EE}[1]{\mathbb{E}\left[{#1}\right]} % Expectation
 
 \newcommand{\EEst}[2]{\mathbb{E}\left[{#1}\ \middle| \ {#2}\right]} % Conditional expectation
 \newcommand{\PPst}[2]{\text{Pr}\!\left\{{#1}\ \middle| \ {#2}\right\}} % Conditional probability
 \def\R{\mathbb{R}}

 \newcommand{\ignore}[1]{}
 \let\emptyset\varnothing

 \newcommand{\nulls}{\mathcal{H}_0}

 \newcommand{\FCR}{\textnormal{FCR}}
 \newcommand{\mFCR}{\textnormal{mFCR}}
 \newcommand{\FCP}{\textnormal{FCP}}
 \newcommand{\FDR}{\textnormal{FDR}}
 \newcommand{\mFDR}{\textnormal{mFDR}}
 \newcommand{\pFCR}{\textnormal{pFCR}}

\newcommand{\FSR}{\textnormal{FSR}}
\newcommand{\mFSR}{\textnormal{mFSR}}
 
 \usepackage{fancyhdr}
 \newcommand{\thedate}{\today}
 \newcommand{\theauthor}{}
 \newcommand{\thetitle}{Online Control of the
 False Coverage Rate and False Sign Rate}
 
 \date{}
 \author{\theauthor}
 \title{\thetitle}

 \fancypagestyle{plain}{%
 	\fancyhf{}%
 	\lhead{\thepage}
 }
 
 \usepackage[mmddyy]{datetime}

 \newcommand{\fcphat}{\widehat{\textnormal{FCP}}}

 \def\S{\mathcal{S}}
 
 \def\R{\mathbb{R}}
 \def\F{\mathcal{F}}

 \def\N{\mathbb{N}}
 \def\X{\mathcal{X}}
 \def\I{\mathcal{I}}

 \makeatletter
 \newcommand{\dotfrac}[2]{
 	\mathchoice
 	{\ooalign{$\genfrac{}{}{0pt}{0}{#1}{#2}$\cr\leavevmode\cleaders\hb@xt@ .22em{\hss $\displaystyle\cdot$\hss}\hfill\kern\z@\cr}}
 	{\ooalign{$\genfrac{}{}{0pt}{1}{#1}{#2}$\cr\leavevmode\cleaders\hb@xt@ .22em{\hss $\textstyle\cdot$\hss}\hfill\kern\z@\cr}}
 	{\ooalign{$\genfrac{}{}{0pt}{2}{#1}{#2}$\cr\leavevmode\cleaders\hb@xt@ .22em{\hss $\scriptstyle\cdot$\hss}\hfill\kern\z@\cr}}
 	{\ooalign{$\genfrac{}{}{0pt}{3}{#1}{#2}$\cr\leavevmode\cleaders\hb@xt@ .22em{\hss $\scriptscriptstyle\cdot$\hss}\hfill\kern\z@\cr}}
 }
 \makeatother

\begin{document}

\author{
Asaf Weinstein, Aaditya Ramdas\thanks{The authors contributed equally to this work.}\\
Department of Statistics and Data Science\\
Carnegie Mellon University\\
\texttt{asafw, aramdas@cmu.edu}
}

\maketitle
\begin{abstract}
The false coverage rate (FCR) is the expected ratio of number of constructed confidence intervals (CIs) that fail to cover their respective parameters to the total number of constructed CIs. Procedures for FCR control exist in the offline setting, but none so far have been designed with the online setting in mind. In the online setting, there is an infinite sequence of fixed unknown parameters $\theta_t$ ordered by time. At each step, we see independent data that is informative about $\theta_t$, and must immediately make a decision whether to report a CI for $\theta_t$ or not. If $\theta_t$ is selected for coverage, the task is to determine how to construct a CI for $\theta_t$ such that $\FCR \leq \alpha$ for any $T\in \N$. A straightforward solution is to construct at each step a $(1-\alpha)$ level conditional CI. 

In this paper, we present a novel solution to the problem inspired by online false discovery rate (FDR) algorithms, which only requires the statistician to be able to construct a marginal CI at any given level. Apart from the fact that marginal CIs are usually simpler to construct than conditional ones, the marginal procedure has an important qualitative advantage over the conditional solution, namely, it allows selection to be determined by the candidate CI itself. We take advantage of this to offer solutions to some online problems which have not been addressed before. For example, we show that our general CI procedure can be used to devise online sign-classification procedures that control the false sign rate (FSR). In terms of power and length of the constructed CIs, we demonstrate that the two approaches have complementary strengths and weaknesses using simulations. Last, all of our methodology applies equally well to online FCR control for prediction intervals, having particular implications for assumption-free selective conformal inference.
\end{abstract}

\section{Introduction}\label{sec:intro}
While statisticians are trained to be aware of multiple testing issues, temporal multiplicity is often easy to miss. Let us examine the following simplified situation alluded to in the abstract. Consider a team of statisticians at a pharmaceutical company who test a new drug every week of the year. In week $t$, a new drug is under consideration, and  to assess its treatment effect $\theta_t$, the team conducts a new randomized clinical trial with new participants. Suppose that the data, such as the normalized empirical difference in means between the treatment and control groups, can be summarized by the observation $X_t \sim N(\theta_t, 1)$, independent of all the previous $X_i$. 

Now consider the following selection rule: if $X_t < 3$, then the statisticians simply ignore drug $t$, and if $X_t > 3$, then the team reports the two-sided $99\%$ marginal CI for $\theta_t$ to the management (who may then decide to run a much larger second phase clinical trial since the CI does not contain 0). This may initially seem like an innocuous situation: each drug is different and has a different treatment effect $\theta_t$,  the data $X_t$ is always fresh and independent, the decision for whether or not to construct the CI for $\theta_t$ is dependent only on $X_t$ and independent of all other $X_i$, and so is the interval if constructed. 

Nonetheless, the combination of multiplicity and selection is a cause for concern already in the offline setting, as was insightfully pointed out by \citet{benjamini2005false}. 
In the online case, when there is an infinite sequence of parameters, it is even easier to construct an example where ignoring selection has undesirable consequences. 
Indeed, consider the special case where $\theta_t=0$ for all $t$, in other words, every tested drug is equivalent to a placebo. 
In this situation, every single CI that is reported to the management is incorrect, since it does not contain zero. 
Because a selection will eventually occur, among constructed CIs the proportion of non-covering CIs---this is later formally defined as the false coverage proportion, FCP---will equal one from this point on. 
Thus, the FCR---expectation of FDP---is not controlled. 
Of course, the second phase of the trial will rectify this error, but at a huge cost of time and money, and loss of faith in the team of statisticians.

One natural solution for this is provided by conditional post-selection inference: instead of a $99\%$ marginal CI, we may construct a conditional $99\%$ interval, where we condition on the event that $X_t > 3$, leading to inference based on a truncated gaussian likelihood in the above setting. 
Confidence intervals based on a truncated normal observation were proposed by \citet{zhong2008bias} and \citet{weinstein2013selection} to counteract the selection effect when providing inference after hypothesis testing. 
While these works consider the batch (offline) setting, in our simple example constructing such conditional CIs (as well as the selection rule) is a legitimate {\it online} CI procedure. 
Furthermore, this controls the FCR---
%Furthermore, as noted in \citet{weinstein2013selection}, this would immediately control the FCR: each constructed CI now individually has nominal coverage, so we are effectively facing multiplicity but no selection, in which case the expected proportion of noncovering CIs is no more than $\alpha m/m=\alpha$, where $m$ is the number of selections made so far. 
in fact, as will be discussed in Section~\ref{sec:conditional} and demonstrated in our simulations, constructing conditional intervals provides unnecessarily strong guarantees, that come at a price. 

%However, it comes at a fairly big cost --- the intervals reported to the management will \emph{often} contain zero \footnote{This can lead to an interesting recursive problem: suppose the team decides to only report those drugs whose \emph{conditional} CIs don't contain zero, then one would need to further correct the conditional CIs by conditioning on this new event, and so on recursively. However, this is true for any conditional post-selection inference method much beyond this paper, and hence we do not discuss this practical issue further here.}.

\smallskip
In this paper, we will propose a new approach for online FCR control that is very different from the aforementioned conditional approach. 
Informally, in order to achieve FCR control at level $\alpha$, instead of constructing $(1-\alpha)$ conditional CIs, we construct $(1-\alpha_i)$ marginal CIs for some $\alpha_i < \alpha$.  
The algorithm to set the $\alpha_i$s is inspired by recent advances in the online false discovery rate (FDR) control literature, specifically recent work by the first author \citep{RYWJ17}. 
The new CI procedure works in much more generality than the simple example described above, that is when $\theta_i$ are multi-dimensional, the data is not necessarily gaussian, and so on---cases in which constructing a conditional CI may be substantially harder if at all possible. 

Even more importantly, by constructing marginal instead of conditional CIs, we leave open the possibility to use as a criterion for selection the candidate CI itself. 
For example, the rule may entail constructing the candidate marginal CI only if it does not include values of opposite signs. 
%specify selection rules that, for example, entail reporting the candidate (marginal) CI only if it excludes zero. 
%Let us reiterate the practical importance of an FCR-controlling procedure with the aforementioned property: r
Thus, returning to the motivating example, this allows the team of statisticians to ensure that each reported CI is conclusive about the {\it direction} of the treatment effect, while the FCR is controlled. 
%Indeed, in most realistic situations the management would not want to see intervals that cross zero, because this reflects ambiguity about the direction of the effect of the corresponding treatments (of course, FCR is not controlled if we simply discard the intervals containing zero). 
With such situations in mind, we instantiate our marginal CI procedure to propose a confidence interval-driven procedure, that constructs sign-determining CIs and can be seen as an online adaptation of the ideas of \citet{weinstein2014selective}. 
Every such sign-determining CI procedure corresponds to an online sign-classification procedure that controls the false sign rate (FSR). 
As a special case we show that for some recently proposed online testing procedures, supplementing rejections based on two-sided $p$-values with directional decisions suffices to control the FSR. 

%In that sense our methods are a generalization of online FDR controlling procedures. 

%Still, because the simple gaussian sequence example described above is convenient for simulations, extensive experiments in this setting suggest that the new approach has better sign-determination, meaning that the constructed intervals cover zero less often.

%Motivated by the aforementioned examples, we also consider a different error metric called the false sign rate (FSR), where the aim is to confidently report the sign of $\theta_t$. We adapt the aforementioned marginal method to produce \emph{asymmetric} sign-determining confidence intervals whenever possible. This new procedure can be seen as the online analog of sign-determining offline FCR procedures studied recently by the second author.

\smallskip
The rest of this paper is organized as follows. 
Section \ref{sec:setup} sets up the problem formally and introduces necessary notation. 
In Section \ref{sec:conditional} we discuss a conditional solution to the online FCR problem. 
A new online procedure that adjusts marginal confidence intervals, is presented in Section \ref{sec:lordci}. 
In Section \ref{sec:localization} we show how our marginal CI procedure can be used to solve a general online localization problem, and study the special case of the online sign-classification problem. 
Simulation results for comparing the marginal approach and the conditional approach are reported in Section \ref{sec:experiments}. We end with a brief discussion in Section \ref{sec:discussion}, where we mention how all of our results also hold for prediction intervals for unseen responses, with further details furnished in Appendix~\ref{sec:conformal}.

\section{Problem Setup}\label{sec:setup}

Let $\theta_1, \theta_2, \dots$ be a fixed sequence of fixed unknown parameters, where the domain $\Theta_i$ of $\theta_i$ is arbitrary, but common examples may include $\R$ or $\R^d$.
Let $2^{\Theta_i}$ denote the set of all measurable subsets of $\Theta_i$, in other words it is any acceptable confidence set for $\theta_i$.  
In our setup, at each time step $i$, we observe an independent observation (or summary statistic) $X_i \in \X_i$, where the distribution of $X_i$ depends on $\theta_i$ (and possibly other parameters). 
For example when $\Theta_i=\R$, we may have $X_i \sim N(\theta_i,1)$. 
Let $\S_i: \X_i \to \{0,1\}$ denote the selection rule that indicates whether or not the user wishes to report a confidence set for $\theta_i$. 
Explicitly, letting $S_i := \S_i(X_i)$ be the indicator for selection, where $S_i=1$  means that the user will report a confidence set for $\theta_i$. 
Let the filtration formed by the sequence of selection decisions be denoted by
$$
\F^i=\sigma(S_1,S_2,\dots,S_i).
$$
% Let $\alpha_i \in [0,1]$ denote the level of coverage that we desire for $\theta_i$. 
Next, let $\I_i: \X_i \times [0,1] \to 2^{\Theta_i}$ be the rule for constructing the confidence set for $\theta_i$, the second argument allowing to take as an input a ``confidence level".   
We denote  $I_i := \I_i(X_i,\alpha_i)$. 
% Thus, for example, if $\alpha_i \in [0,1]$ represents the confidence level, then $I_i$ represents a $(1-\alpha_i)$ confidence set, with more details to be furnished later. 
Thus, $I_i = \I_i(X_i,\alpha_i)$ may be a marginal or a conditional $(1-\alpha_i)$ confidence set for $\theta_i$ as discussed later, but in general it is no more than a map from $\X_i \times [0,1]$ as described above.
% Let $\mathbb{I}_i$ denote the set of all such rules, that is the space in which $\I_i$ lives.
% but we emphasize that it is in general no more than a function that takes as input a pair $(x_i,a)$, with $x_i\in \X_i, a\in [0,1]$. 
%, meaning that $\I_i(x,a)$ is a $(1-a)$-level confidence set for $\theta_i$ after observing $x$.  
For simplicity, in the rest of the paper we refer to $I_i$ as a confidence \emph{interval} (CI) like it would usually be if $\Theta_i=\R$, but with the understanding that everything discussed in this paper applies to the more general case of arbitrary confidence sets.

% \awcomment{I strongly object to defining $\S_i$ as a function of both $X_i$ and $I_i$... we can talk about it over the phone. Nothing is compromised if we leave the definition as it was before, i.e., $\S_i: \X_i \to \{0,1\}$, as long as we require in the sequel that $\S_i$ be predictable... Am I wrong?}
% and everything that follows immediately generalizes to the multivariate or discrete or even more general settings.
%but with the understanding that nothing is specific to the univariate real-valued setting and everything that follows immediately generalizes to the multivariate or discrete or even more general settings.

In our setup, the above rules are all required to be {\it predictable}, meaning that
\[
\S_i, \I_i, \alpha_i \text{ are } \F^{i-1}\text{-measurable},
\]
and we write $\S_i, \I_i, \alpha_i \in \F^{i-1}$. 
Naturally, the instantiated random variables $I_i=\I_i(X_i,\alpha_i), S_i=\S_i(X_i)$ both depend on $X_i$.
 % and hence are all adapted to $\F^i$, meaning they are $\F^i$-measurable. 
However, the \emph{rules} $\S_i,\I_i,\alpha_i$ must be $\F^{i-1}$-measurable, hence specified before observing $X_i$. 
We emphasize that the requirement to be $\F^{i-1}$-measurable also prevents the rules $\S_i, \alpha_i, \I_i$ from depending on $(X_1,\dots,X_{i-1})$ unless it is through $(S_1,\dots,S_{i-1})$. Importantly, $\S_i$ can depend on $\I_i$ because both are predictable, and hence $S_i$ can depend on $I_i$---for example, whether or not $I_i$ looks ``favorable'', a point to which we will return in later sections.

\medskip
Using these definitions, we now define an online selective-CI procedure. 
In the rest of the paper, we omit the term ``selective", but this is done only for the sake of readability. 
Thus, an {\it online CI protocol} proceeds as follows:

% In other words, a monotone selection rule is only more likely to wish to cover more extreme values. 
% This means that the selection rules must choose some finite, non-negative random variables $0 \leq \sigma_i,\tau_i \in \F^{i-1}$ and be of the form:
% \begin{align*}
% \text{(everything) }\quad & \S_i(x) = 1, \text{ or } \\
% \text{(one-sided, +) }\quad & \S_i(x) := \one{x \geq \tau_i}, \text{ or }\\
% % \text{(one-sided, -) }\quad & \S_i(x) := \one{x \leq -\sigma_i}, \text{ or }\\
% \text{(two-sided, symmetric) }\quad & \S_i(x) := \one{|x| \geq \tau_i}, \text{ or }\\
% \text{(two-sided, asymmetric) }\quad & \S_i(x) := \max\{\one{x \geq \tau_i}, \one{x \leq -\sigma_i}\}
% \\
% \text{(nothing) }\quad & \S_i(x) = 0,
% \end{align*}
% or be minor variants in which $\geq$ is replaced by $>$. (Without loss of generality, we do not need to consider the one-sided negative case, because it can be simply converted to the one-sided positive case by negating both $X_i$ and $\theta_i$.)

% The monotonicity assumption on each $\S_i$ is not an unreasonable one. For example, the first two selection rules are natural if $X_i$ represents the observed treatement effect in a clinical trial (or an A/B test in the tech industry). If one is interested in estimating only strong effects in any direction, then all but the last option are relevant. The first and last options are included just as trivial special cases, but are obviously the least interesting from a selective inference viewpoint.

\begin{enumerate}
\item At time $i$, first commit to $\alpha_i, \S_i, \I_i \in \F^{i-1}$.
\item Then, observe $X_i$. 
% \footnote{We intentionally use the word ``candidate", because we will later want to refer to $I_i$ even if it is ultimately not constructed. Note that $\I_i\in \F^{i-1}$ must be well-defined when $\PP{S_i=1}>0$. If $\I_i$ is not well defined (for example, if we construct conditional CIs and the rule is ``never report an interval for $\theta_3$", there is practically no need to make a choice for $\I_3$), then simply set this to NA.}  
% Construct the candidate interval $I_i=\I_i(X_i,\alpha_i)$. 
% \awcomment{It's not natural to say that you form $I_i$ before deciding if it will be reported: for example, if you report whenever $X_i>3$, there is certainly no reason to contemplate what the CI would look like when $X_i<3$... I'm trying to say that this may confuse the reader...}
Decide whether or not $\theta_i$ is selected for coverage by setting $S_i=\S_i(X_i)$. 
\item Report $I_i=\I_i(X_i,\alpha_i)$ if $S_i=1$. Then, increment $i$, and go back to step 1.
\end{enumerate}
We next discuss the metrics used to evaluate the errors made by an online CI protocol.

\subsection{Error metrics}

Let the unknown false coverage indicator be denoted $V_i := S_i \one{\theta_i \notin I_i}$. 
Hence, $V_i=1$ implies that we intended to cover $\theta_i$ but our reported CI $I_i$ failed to do so. Using the aforementioned terminology, define the false coverage proportion up to time $T$ as 
\[
\FCP(T) = \frac{\#\text{reported intervals that fail to cover their parameter}}{\#\text{reported intervals}} = \frac{\sum_{i \leq T} V_i}{\sum_{i \leq T} S_i},
\]
where $0/0=0$ per standard convention (i.e., if no intervals are constructed, then the false coverage proportion is trivially zero). 
The false coverage rate (FCR) and the modified FCR are defined, respectively, as
\[
\FCR(T) = \EE{\frac{\sum_{i \leq T} V_i}{\sum_{i \leq T} S_i}}, \ \ \ \ \  \mFCR(T) = \frac{\EE{\sum_{i \leq T} V_i}}{\EE{\sum_{i \leq T} S_i}}.
\]
Along the way, we will consider the relationship of the FCR to other error metrics like the positive FCR (pFCR), the false sign rate (FSR) and the well-known false discovery rate (FDR).
% \awcomment{I object to this entire sentence: (i) if you mention pFCR here, then why not define it? (ii) why mix pFCR and FSR? these are not related to each other... (iii) It's true that ``FCR control `covers' the cases of FSR and FDR control" through ``LORD-CI for localization", but the reader may not realize this at this early point in the paper.. Also, I wouldn't use ``special cases" anyway...}

% In the hypothesis testing setting, it has been repeatedly observed that there is often little practical difference between the FDR and mFDR, so here as well we will be okay with methods that control either one. 

\subsection{Main objective}

The main objective of this paper is to develop and compare algorithms to specify $\I_i$ and $\alpha_i$ such that FCR or mFCR control is guaranteed at any time regardless of the choice of $\S_i$, that is, 
\[
\FCR(T) \leq \alpha \ \ \ \forall T \in \N, \ \ \ \ \ \ \text{ or }  \ \ \ \ \ \  \mFCR(T) \leq \alpha \ \ \ \forall T \in \N. 
\]
Specifically, we explore the following two avenues for constructing the CIs:
\begin{enumerate} 
\item Marginal CI: this has the guarantee that for any $a \in [0,1]$, we have
\begin{equation}
\label{def:marg-CI}
\PPst{\theta_i \notin \I_i(X_i,a)}{\F^{i-1}} \leq a,
\end{equation}
where the probability is taken only over the marginal measure of $X_i$, because the rule $\I_i$ is predictable.
\item Conditional CI: this has the property that for any $a \in [0,1]$, we have
\begin{equation}
\label{def:cond-CI}
% \PPst{\theta_i \notin \I_i(X_i,a)}{S_i=1, S_{i-1}=s_{i-1},\dots,S_1=s_1} \leq a \ \ \ \ \ \ \ \forall (s_{i-1},\dots,s_1),
\PPst{\theta_i \notin \I_i(X_i,a)}{\F^{i-1},S_i=1} \leq a,
\end{equation}
where the probability is taken over the measure of $X_i$ conditional on $S_i=1$, because $\I_i$ is predictable. \footnote{
In defining a conditional CI, one may consider requiring only that $\PPst{\theta_i \notin \I_i(X_i,a)}{S_i=1} \leq a$. 
This weaker condition will suffice for $\mFCR$ control, as can be seen from the proofs of our theorems. 
We chose to use the stronger requirement \eqref{def:cond-CI} partly because it is more natural to construct a conditional CI when conditioning on $S_1,\dots,S_{i-1}$ along with $S_i$; indeed, our simulations include a typical example where we do not know how to construct a conditional CI with the weaker property, but it is easy to construct one with the stronger property \eqref{def:cond-CI}.
}
\end{enumerate} 
For either choice, we must specify the level $\alpha_i \in (0,1)$ to use with $I_i$ if $\theta_i$ is selected for coverage. 

\smallskip
On accomplishing this main objective, we detail in Section \ref{sec:localization} exactly how it enables us to solve several other practical problems of interest, such as controlling the false sign rate. As mentioned in the end of the discussion in Section \ref{sec:discussion}, the entire setup of this paper applies equally well to prediction intervals instead of CIs.
% \awcomment{I'm not sure this sentence is clear and meaningful enough at this stage...}

\section{A method based on conditional inference}\label{sec:conditional}

A conceptually straightforward method to control the mFCR is to construct conditional CIs at the nominal level $(1-\alpha)$. This trivially controls the mFCR at level $\alpha$, as seen by the following argument.

\begin{theorem}
Constructing a $(1-\alpha)$ conditional CI after every selection ensures that $\forall T \in \N, \mFCR(T) \leq \alpha$.
\end{theorem}

\begin{proof}
From the definition \eqref{def:cond-CI} of a conditional CI it follows immediately that 
\begin{equation*}
\EEst{V_i}{S_i=1} = \EEst{I_{\theta_i\notin I_i}}{S_i=1 } = \PPst{\theta_i\notin I_i}{S_i=1} \leq \alpha. 
\end{equation*}
Together with the fact that $\EEst{V_i}{S_i=0} = 0$, we have 
\begin{equation*}
\EEst{V_i}{S_i} \leq \alpha \ \ \ \ \ \ \ \text{a.s.}, 
\end{equation*}
and hence,
\begin{align*}
\EE{\sum_i V_i} &= \sum_i \EE{V_i} = \sum_i \EE{S_i V_i}\\
&= \sum_i \EE{S_i \EEst{V_i}{S_i}} \leq \sum_i \EE{\alpha S_i}\\
&=\alpha \sum_i \EE{S_i} = \alpha \EE{\sum_i S_i}.
\end{align*}
Rearranging the first and last displays above yields the desired result.
\end{proof}

Constructing conditional CIs at the nominal level ensures also that $\FCR$ is controlled. 
As a matter of fact, even the conditional expectation of $\FCP$ given that at least one selection is made,
$$
\pFCR(T) := \EEst{\FCP(T)}{\sum_{i=1}^T S_i >0},
$$
is controlled when using conditional CIs. 
We call the above the {\it positive \FCR}, in analogy to the positive FDR \citep{storey2003positive}.

\begin{theorem}
Constructing a $(1-\alpha)$ conditional CI after every selection ensures that 
$$
\pFCR(T) \leq \alpha \ \ \ \ \ \ \ \ \forall T \in \N.
$$
\end{theorem}

\begin{proof}
Consider any sequence $(s_1,\dots,s_T)\in \{0,1\}^T$ such that $\sum_i s_i >0$. 
We have
\begin{align*}
\EEst{\frac{\sum_i V_i}{\sum_i S_i}}{S_1 = s_1,\dots,S_T=s_T} &= \frac{1}{\sum_i s_i} \EEst{\sum_i V_i}{S_1 = s_1,\dots,S_T=s_T}\\
&= \frac{1}{\sum_i s_i} \EEst{\sum_{\{i\leq T: s_i = 1\}} I_{\theta_i\notin I_i}}{S_1 = s_1,\dots,S_T=s_T}\\
&= \frac{1}{\sum_i s_i} \sum_{\{i\leq T: s_i = 1\}} \PPst{\theta_i\notin I_i}{S_1 = s_1,\dots,S_T=s_T}\\
&\stackrel{(a)}{=} \frac{1}{\sum_i s_i} \sum_{\{i\leq T: s_i = 1\}} \PPst{\theta_i\notin I_i}{S_1 = s_1,\dots,S_{i-1}=s_{i-1}, S_i=1}\\
&\leq \frac{1}{\sum_i s_i} \sum_{\{i\leq T: s_i = 1\}} \alpha\\
&=\alpha,
\end{align*}
where equality $(a)$ uses the fact that the selection decisions  $(S_{i+1},\dots,S_T)$ are independent of $X_i$ given $(S_1,\dots,S_{i})$ because the selection rules $\{\S_j\}$ are predictable. 
The original claim follows by taking expectation over the conditional distribution of $S_1,\dots,S_T$ given that $\sum_{i=1}^TS_i >0$. 
%which implies what we need by taking expectation over the conditional distribution of $S_1,\dots,S_T$ given that $\sum_{i=1}^TS_i >0$. 
\end{proof}

We immediately conclude that with conditional $(1-\alpha)$ CIs we also have
$$
\FCR(T) = \pFCR(T) \cdot \PP{\sum_{i=1}^T S_i >0}\leq \pFCR(T) \leq \alpha.
$$

Control of the $\pFCR$ (and hence FCR) may seem pleasant, but in fact this strong guarantee has a price. Our two main criticisms of the conditional approach are:
\begin{enumerate}
\item \textbf{Incompatibility.}  Conditional CIs are not able to ensure compatibility between selection decisions and the reported CI. For example, it is impossible to ensure that all selected CIs are sign-determining, meaning that it is impossible to select only those confidence intervals that do not contain 0. This is discussed further and explicitly demonstrated in Subsection \ref{subsec:inconsistency}.

\item \textbf{Intractability.} The conditional distribution of $X_i$ given $(S_{i-1},\dots,S_1)$ and the event $\{S_i=1\}$, is the distribution resulting from restricting $X_i$ to some subset of $\X$, which may be intractable to compute in general. 
At the very least, the conditional approach requires a case-by-case treatment; depending on the marginal distribution of $X_i$ and the selection rules $\S_1,\dots,\S_i$, computing the conditional distribution may be far from trivial. 
\end{enumerate}

In the next sections, we describe a marginal approach to controlling the FCR, and elaborate on its various advantages with respect to the aforementioned conditional approach.

\section{Adjusting marginal intervals: the LORD-CI procedure}\label{sec:lordci}

In what follows, an {\it algorithm} is a sequence of mappings from past selection decisions to
confidence levels, meaning that it maps $(S_1,\dots,S_{i-1})$ to $\alpha_i$. 
By definition, such an $\alpha_i$ is $\F^{i-1}$-measurable, hence a procedure that constructs a {\it marginal} confidence interval at level $(1-\alpha_i)$ whenever $S_i=1$, is a legitimate online CI protocol. 
We will refer to such a procedure as a {\it marginal online CI protocol/procedure}. 
A trivial marginal online CI protocol can be obtained by taking any {\it fixed} sequence of $\alpha_i$ such that the series $\sum_{i=1}^\infty \alpha_i \leq \alpha$; this procedure is called alpha-spending in the context of online FDR control by \cite{foster2008alpha}, and controls the familywise error rate (which in our context is the probability of even a single miscoverage event). Naturally, this is a much more stringent notion of error, and hence the resulting selected CIs will be excessively wide. The question we will address below is the following: is there a nontrivial algorithm to set the $\alpha_i$ so that FCR is controlled? 
% \awcomment{should I take care of completing this?}
%We now describe a new algorithm, which we call LORD-CI, that produces a sequence of $\alpha_i$ levels such that the FCR is controlled at level $\alpha$. It is motivated by the LORD algorithm by \cite{javanmard2016online} and its improvement LORD++ by \cite{RYWJ17}, that were both introduced in the context of the online FDR problem mentioned in the introduction. We first define

\subsection{mFCR control for arbitrary selection rules}

\smallskip
Our first result identifies a sufficient condition for an algorithm to imply $\mFCR$ control. 
Thus, we first associate any algorithm with an estimated false coverage proportion,
\[
\fcphat(T) := \frac{\sum_{i\leq T} \alpha_i}{(\sum_{i \leq T} S_i) \vee 1}.
\]
We may then define the following procedure for online FCR control.

\begin{definition}[LORD-CI procedure] \label{def:lord:ci:gen}
A LORD-CI procedure is any online protocol that constructs marginal $(1-\alpha_i)$ confidence intervals, where $\alpha_i \in \F^{i-1}$ are defined in a predictable fashion to maintain the invariant
\begin{equation}\label{eq:LORDCI-invariant}
\forall T \in \N, \quad \fcphat(T) \leq \alpha,
\end{equation}
regardless of the selection rules $\S_i$.
\end{definition}

\noindent Any LORD-CI procedure comes with the following theoretical guarantee.

\begin{theorem}
Given an arbitrary sequence of selection rules made by the user, any LORD-CI procedure has the guarantee that $\forall T \in \N, \mFCR(T) \leq \alpha$.
\end{theorem}

%The main idea behind LORD-CI is that the levels $\alpha_i$ will be chosen such that the following invariant will hold:
%\begin{equation}\label{eq:LORDCI-invariant}
%\forall T \in \N, \quad \fcphat(T) \leq \alpha.
%\end{equation}
%We will maintain this invariant by running the aforementioned LORD++ online FDR algorithm. While LORD++ uses rejection decisions based on $p$-values to update  $\alpha_i$, in the absence of $p$-values, LORD-CI instead substitutes these rejection events by selection events ($S_i=1$). This algorithm comes with the following guarantees.

\begin{proof}
By definition of a false coverage event, we have
\begin{align*}
\EE{\sum_i V_i} 
&=~ \EE{\sum_i S_i I_{\theta_i \notin I_i}} 
 ~\stackrel{(a)}{\leq}~ \sum_i \EE{I_{\theta_i \notin I_i}} \\
 &=~ \sum_i \EE{\EEst{I_{\theta_i \notin I_i}}{\F^{i-1}}} 
 ~\stackrel{(b)}{\leq}~ \sum_i \EE{\alpha_i} \\
 &=~ \EE{\sum_i \alpha_i} 
 ~\stackrel{(c)}{\leq}~ \alpha \EE{\sum_i S_i},
\end{align*}
where  inequality $(a)$ holds because $S_i \leq 1$, inequality $(b)$ by the definition \eqref{def:marg-CI} of a $(1-\alpha_i)$ marginal CI, and inequality $(c)$ by the invariance \eqref{eq:LORDCI-invariant}. Rearranging the first and last expression yields the desired result.
\end{proof}

If one really insisted on requiring FCR control as opposed to mFCR control, we provide a guarantee for a subclass of ``monotone'' selection rules, as introduced below. 

\subsection{Monotonicity of algorithms, intervals and selection rules}

The symbol $\succeq$ is used to compare vectors coordinatewise, so $(v) \succeq (w)$ means that $v_i \geq w_i$ for all $i$. 

An online FCR algorithm is called \emph{monotone} if for any two vectors $(s_1,\dots,s_{i-1}) \succeq (\widetilde s_1, \dots, \widetilde s_{i-1})$, we have $\alpha_i(s_1,\dots,s_{i-1}) \geq \alpha_i(\widetilde s_1,\dots, \widetilde s_{i-1})$. 
Equivalently, an online FCR algorithm is monotone if
\begin{equation}
\alpha_i \geq \widetilde{\alpha}_i \ \ \ \text{ whenever } (S_1,\dots,S_{i-1}) \succeq (\widetilde S_1, \dots, \widetilde S_{i-1}),
\end{equation}
where $\widetilde \alpha_i$ is the level produced by the online FCR algorithm, when presented with the history of selection decisions $(\widetilde S_1, \dots, \widetilde S_{i-1})$.
% An online FCR algorithm is called \emph{monotone} if 
% \[
% \alpha_i \geq \widetilde \alpha_i \text{ whenever } (S_1,\dots,S_{i-1}) \succeq (\widetilde S_1, \dots, \widetilde S_{i-1}).
% \]
% This is not really a restriction, because the algorithm design is under our control, and there are simple rules to ensure monotonicity of $\alpha_i$.
We say that a CI rule $\I$ is {\it monotone} if 
\[
\I(x,a_2)\subseteq \I(x,a_1) \text{ for all } a_1< a_2 \text{ and } x\in \X. 
\]
Monotonicity is satisfied for most natural (even non-equivariant) CI constructions, and thus we do not view this as a restriction. Irrespective of whether the online FCR algorithm and CI rule are monotone, we say that a selection rule $\S_i$ is \emph{monotone} if 
\begin{equation}
% \S_i(x) \geq \S_i(y) \text{ whenever } x \geq y \geq 0 \text{ or } x \leq y \leq 0 \\
% \text{ and } 
S_i \geq \widetilde S_i \text{ whenever } (S_1,\dots,S_{i-1}) \succeq (\widetilde S_1, \dots, \widetilde S_{i-1}), \label{eq:monotone}
\end{equation}
where, as before, $\widetilde S_i$ is used to denote the selection decision at time $i$, for the same observation $X_i$, but for a different history $(\widetilde S_1, \dots, \widetilde S_{i-1})$.

As a simple special case, if each rule $\S_i$ is independent of $\F^{i-1}$, then such a selection rule is trivially monotone, even if the underlying online FCR algorithm is not. 
In other words, if the final decision $S_i=\S_i(X_i)$  is based only on $X_i$ and on none of the past decisions, then such a rule is monotone. 
For example, setting $S_i = \one{X_i > 3}$ for every $i$ constitutes a trivial monotone selection rule. 

% For a slightly more sophisticated example, consider any selection rule $\S_i$ whose sole dependence on $\F^{i-1}$ is through a monotone CI rule $\I_i$, such that $\S_i$ always prefers shorter intervals, meaning that $S_i \geq \widetilde S_i$ whenever $I_i \subseteq \widetilde I_i$. For any such rule, if a monotone online FCR algorithm is employed to generate $\alpha_i$, then the selection rule $\S_i$ also becomes monotone. In brief, monotonicity of $\alpha_i$ and $\I_i$ together result in the monotonicity of any rule $\S_i$ that depends on $\F^{i-1}$ only via $\I_{i}$.

% \awcomment{You will probably disagree, but I don't think we should go into ``dependence on $\F^{i-1}$ through a monotone CI rule $\I_i$"... let's discuss over phone}
% Otherwise, the monotonicity requirement mandates that selecting a variable for coverage cannot make selection of a future variable \emph{less} likely; in other words, increasing interest in one variable cannot decrease interest in later ones. While it would be certainly nice to drop this restriction (and indeed mFCR control does not require it), 

\subsection{FCR control for monotone selection rules}

We can provide the following guarantee for the nontrivial class of monotone selection rules. 

\begin{theorem}\label{thm:lord:ci:monotone}
Given an arbitrary sequence of monotone selection rules chosen by the user, 
any LORD-CI procedure that maintains the invariant \eqref{eq:LORDCI-invariant} also satisfies that $\forall T \in \N, \FCR(T) \leq \alpha$.
\end{theorem}

\begin{proof}
By definition of $\FCR(T)$, we have
\begin{align}
\FCR(T) &= \EE{\frac{\sum_{i \leq T} V_i}{\sum_{j \leq T} S_j}}
~=~ \sum_{i \leq T} \EE{ \frac{S_i \one{\theta_i \notin I_i} }{\sum_{j \leq T} S_j}} ~\leq~ 
\sum_{i \leq T} \EE{ \frac{\alpha_i }{\sum_{j \leq T} S_j}},
\end{align}
where the sole inequality follows by Lemma \ref{lem:transfer}, introduced after this proof. 
Thus, we see that
\[
\FCR(T) \leq \EE{ \frac{ \sum_{i \leq T} \alpha_i }{\sum_{j \leq T} S_j} } \leq \alpha,
\]
where the last inequality holds due to invariant \eqref{eq:LORDCI-invariant}.
\end{proof}

The critical step in the aforementioned argument is the invocation of the following powerful lemma.

\begin{lemma}\label{lem:transfer}
Given an arbitrary sequence of monotone selection rules, we have
\[
\EE{ \underbrace{\frac{S_i \one{\theta_i \notin I_i} }{\sum_{j \leq T} S_j}}_{A_i} }
\leq 
\EE{ \frac{\alpha_i }{\sum_{j \leq T} S_j} }.
\]
\end{lemma}
Intuitively, the statement of the above lemma is obvious if the expectation could be taken separately in the numerator, as if it was independent of the denominator, because $\EE{S_i\one{\theta_i \notin I_i}} \leq \EE{\one{\theta_i \notin I_i}} \leq \alpha_i$ by construction \eqref{def:marg-CI}. The following proof demonstrates that monotonicity allows us to formally perform such a step.
\begin{proof}
Without loss of generality, we can ignore the case when $S_i=0$ almost surely for some $i$; in other words, if we would never select $\theta_i$, then $V_i=0$ almost surely, and we can just ignore the time instant $i$. Hence, we only consider the case when at least one value of $X_i$ leads to selection.

To derive a bound on $\EE{A_i}$, consider the following thought experiment. 
Let us hallucinate what selection decisions would have occurred under a slightly different series of observations, namely
\[
\widetilde{X} := (X_1,X_2,\dots,X_{i-1},X^*,X_{i+1},\dots,X_T),
\]
% Note that instead of $+\infty$, one can just use the maximum possible observation in the domain of $X_i$, or the maximum numerical value that could be stored on a computer.
% The only thing that matters is that $\widetilde{X}_i$ is non-random, and since the selection decisions are monontone, we have $\widetilde S_i=1$ independently of $\F^{i-1}$. 
where $X^*$ is any value that would have led to selection of $\theta_i$, which is a predictable choice, because it can be made based on only the predictable selection rule $\S_i$. Let the sequence of selection decisions made by the same algorithm  on $\widetilde X$ be denoted $\widetilde S_i$, the levels be denoted $\widetilde \alpha_i$, and the constructed intervals be $\widetilde I_i$.
We then claim that
\[
A_i \equiv \frac{S_i \one{\theta_i \notin I_i} }{\sum_{j \leq T} S_j} = \frac{S_i\one{\theta_i \notin I_i} }{\sum_{j \leq T} \widetilde S_j} =: \widetilde{A}_i,
\]
where we have intentionally altered only the denominator.
To see that the above equality holds, first note that if $S_i=0$, then $A_i = \widetilde{A}_i=0$. 
Then note that if $S_i=1$, then $\widetilde{S}_i=S_i$ for all $i$. 
Indeed, because  $X_j = \widetilde{X}_j,\text{ for } j\leq i-1$, the first $i-1$ selection decisions are identical by construction; then if $S_i=1$ (and $\widetilde S_i=1$ by construction), then $\F^i = \widetilde{\F}^i$, and so every future selection decision is also identical (and also the constructed CIs, at levels $\alpha_i$). Hence,
\begin{align*}
\EE{A_i} &= \EE{\widetilde{A_i}} \stackrel{(a)}{\leq} \EE{ \frac{\one{\theta_i \notin I_i} }{\sum_{j \leq T} \widetilde S_j} }\\
&\stackrel{(b)}{=} \EE{\frac1{\sum_{j \leq T} \widetilde S_j} \EEst{ \one{\theta_i \notin I_i}  }{\widetilde{\F}^{n \backslash i}} } \\
&\stackrel{(c)}{\leq} \EE{ \frac{ \alpha_i }{\sum_{j \leq T} \widetilde S_j} } ~\stackrel{(d)}{\leq}~ \EE{ \frac{ \alpha_i }{\sum_{j \leq T} S_j} },
\end{align*}
where inequality $(a)$ holds because $S_i \leq 1$, equality $(b)$ follows because $\frac1{\sum_{j \leq T} \widetilde S_j} $ is $\widetilde{\F}^{n \backslash i}$-measurable because $\widetilde S_i = 1$ by construction, inequality $(c)$ holds by definition \eqref{def:marg-CI} of a marginal CI, and inequality $(d)$ holds because $\widetilde S_j \geq S_j$ for all $j$ by the monotonicity of selection rules. This completes the proof of the lemma.
\end{proof}

The above is a generalization of lemmas that have been proved in the context of online FDR control by \cite{javanmard2016online,RYWJ17}, since the selection event $\{S_i=1\}$ may or may not be associated with the miscoverage event $\{\theta_i \notin I_i\}$, but in online FDR control, the rejection event $\{R_i=1\}\equiv\{P_i \leq \alpha_i\}$ is obviously directly related to the false discovery event $\{P_i \leq \alpha_i,i \in \nulls\}$. We will later see that online FCR control captures online FDR control as a special case.

\subsection{An explicit monotone online FCR algorithm}

By the theorems above, the class of procedures in Definition \ref{def:lord:ci:gen} yield mFCR (FCR) control. 
% when used with any marginal online CI protocol that maintains the invariant \eqref{eq:LORDCI-invariant}. 
%Note that, when the event $\{S_i=1\}$ corresponds to rejection of the $i$-th null hypothesis, the condition \eqref{eq:LORDCI-invariant} is satisfied by some online FDR procedures. 
To obtain a specific procedure, we use the LORD++ online FDR algorithm \citep{RYWJ17} to set the sequence of the $\alpha_i$. LORD++ was originally designed to maintain the invariant \eqref{eq:LORDCI-invariant} in the context of testing, i.e., when $S_i=1$ stands for rejection of the $i$-th null hypothesis. 
In the absence of $p$-values, our algorithm instead substitutes rejection events by arbitrary selection events ($S_i=1$). 
We call the aforementioned adaptation of LORD++ to the context of CIs {\it the LORD-CI algorithm}, and refer to the corresponding marginal online CI protocol as {\it the LORD-CI procedure}. 
In the sequel, unless indicated otherwise, whenever we refer to a LORD-CI procedure (or simply {LORD-CI}), we mean {\it the} LORD-CI procedure, that is, the protocol utilizing LORD++.
An explicit description of LORD-CI is given in Protocol \ref{alg:lordci} below. 

\smallskip
In Protocol \ref{alg:lordci}, $\{\gamma_i\}_{i=1}^{\infty}$ is a deterministic nonincreasing sequence of positive constants summing to one, that is specified in advance; $0\leq W_0\leq \alpha$ is a prespecified constant; $\{\S_i\}$ is a sequence of arbitrary predictable selection rules; and $\{\I_i\}$ is now a sequence of marginal CIs, that is, each $\I_i$ has the property \eqref{def:marg-CI}. 
On implementing Protocol \ref{alg:lordci}, set $\gamma_i = 0$ whenever $i\leq 0$ (this happens for $j=1$). 
% Very briefly, we begin with an error budget of $W_0$, our budget decreases by $\alpha_i$ at every single step, and increases by $\alpha-W_0$ at the first selection and by $\alpha$ at every future selection. 
It is easy to verify that $\alpha_i$ (see line 7 in Protocol \ref{alg:lordci}) is monotone, because it has an additional nonnegative term in the summation with every new selection. One may also verify that $\{\alpha_i\}$ satisfies the invariant \eqref{eq:LORDCI-invariant}, because $\sum_i \alpha_i$ is always less than $(\sum_i S_i)\alpha$.

%\arcomment{Explicitly describe the LORD-CI procedure, without assuming readers will go look up earlier work}

\bigskip

%\begin{minipage}{1\linewidth}

\begin{algorithm}[H]
\SetAlgoLined
\SetKwInOut{Input}{Input}
\SetKwInOut{Output}{Output}
\Input{sequence $\{X_i\}$ observed sequentially; \ prespecified deterministic sequence $\{\gamma_i\}$ summing to one; \ constant $\  W_0 \in (0,\alpha)$; \ arbitrary selection rules $\{\S_i\}$; \ marginal CI rules $\{\I_i\}$; \ $\alpha$}
\Output{online FCR-adjusted selective CIs}
$i \leftarrow 1$ \hspace{3.5in} // tracks time\\
\For{$j=1,2,...$}{
\While{$\S_i(X_i)=0$}{
$i \leftarrow i+1$
}
$\tau_j \leftarrow i$ \hspace{3.25in} // time of the $j$-th selection \\
$
\alpha_i \leftarrow \gamma_iW_0 + (\alpha-W_0)\gamma_{i-\tau_1} + \alpha \displaystyle \sum_{\{k: \tau_{k}<i, \tau_{k}\neq \tau_1\}}\gamma_{i-\tau_{k}} 
$ \quad \quad \quad \quad // a monotone update rule\\
Report $I_i = \I_i(X_i,\alpha_i)$\\
$i \leftarrow i+1$
}

\caption{The LORD-CI procedure}
\label{alg:lordci}
\end{algorithm}

%\begin{algorithm}[H]
%\SetAlgoLined
%\SetKwInOut{Input}{Input}
%\SetKwInOut{Output}{Output}
%\Input{A sequence $X_1,X_2,...$ observed sequentially; $\ W_0; \ \{\gamma_i\}$}
%\Output{FCR-adjusted selective CIs}
%$i \leftarrow 1$\;
%\While{$\S_i(X_i)=0$}{
%$i \leftarrow i+1$\;
%}
%$\tau_1 \leftarrow i$\;
%$\alpha_i \leftarrow \gamma_iW_0$\;
%Report $I_i = \I_i(X_i,\alpha_i)$\;
%$i \leftarrow i+1$
%
%\While{$\S_i(X_i)=0$}{
%$i \leftarrow i+1$\;
%}
%$\tau_2 \leftarrow i$\;
%$\alpha_i \leftarrow \gamma_iW_0 + (\alpha-W_0)\gamma_{i-\tau_1}$\;
%Report $I_i = \I_i(X_i,\alpha_i)$\;
%$i \leftarrow i+1$
%
%\For{$i=3,4,...$}{
%\While{$\S_i(X_i)=0$}{
%$i \leftarrow i+1$\;
%}
%$\tau_2 \leftarrow i$\;
%$\alpha_i \leftarrow \gamma_iW_0 + (\alpha-W_0)\gamma_{i-\tau_1} + \alpha \sum_{j: \tau_j<i, \tau_j\neq \tau_1}\gamma_{i-\tau_j}$\;
%Report $I_i = \I_i(X_i,\alpha_i)$\;
%$i \leftarrow i+1$
%}
%
%\caption{The LORD-CI Procedure}
%\end{algorithm}

% \end{minipage}

% \smallskip
% Another example of an online FDR algorithm satisfying \eqref{eq:LORDCI-invariant} is\dots.

% \noindent \awcomment{I would also make a comment about online FDR procedures in general: it makes sense to ask whether using in a CI procedure the levels $\alpha_i$ output by any online FDR procedure that is predictable in the sense of the current paper, controls the mFCR. 
% I don't know if this is true or not. 
% But I would anyway comment on that..
% }

\section{Selections that depend on the candidate CIs}\label{sec:localization}

In the LORD-CI procedure, the predictable sequences of selection rules $\S_i$ and marginal CI rules $\I_i$ are both arbitrary, and these may be specified independently of each other. 
In this section we demonstrate how tying the selection rule to the confidence interval rule, by letting the (candidate) marginal CI {\it determine} whether $\theta_i$ is selected or not, can lead to many instantiations of the LORD-CI procedure that are of practical interest. 

Informally, the idea is as follows. 
Suppose that we made some choice in advance for the marginal CI rule. 
Suppose also that we have in mind a criterion for what constitutes a ``good" reported CI. 
For example, when $\Theta_i\equiv \R$, we might consider a reported CI ``good" if it excludes zero. 
%More generally, the criterion is defined through a pre-specified subset of $\X_i\times [0,1]$ (in the example, this is $\{(x,a):0\notin \I(x,a)\}$). 
Then at each step $i$, upon observing $X_i$, we pretend that we were to construct $I_i = \I(X_i,\alpha_i)$ where $\alpha_i$ are set by the LORD-CI algorithm, but we only actually select and report it if it is ``good". 
By design, then, we only report ``good" intervals. 
Note that because $\I_i$ is predictable, choosing to report an interval only if it is ``good'' is a {\it predictable} selection rule. Therefore we may use LORD-CI to determine levels of coverage and immediately be guaranteed FCR control. 
This is formalized in the definition below. 
We just remark that these ideas appeared first in \citet{weinstein2014selective}, but their treatment is rather informal, and, importantly, their proposed procedure is not an online procedure. 

For the remainder of this section, whenever we speak of a CI rule, it will be assumed to be monotone. Again, we do not view this as a restriction. 
% (If the CIs are not monotone, the error metric will correspond to a modified FCR, like the mFCR that takes ratios of expectations.)

\subsection{From coverage to localization}
Suppose that for each $i$ we have a collection $K_{i1},\dots,K_{iL_i} \in 2^{\Theta_i}$ of pre-specified disjoint subsets of $\Theta_i$. Being able to say that $\theta_i\in K_{il}$ for exactly one $l\in \{1,\dots,L_i\}$ qualifies as having ``localized'' the signal \footnote{If the sets are not disjoint a-priori, one may either create a new set for the intersection, or generalize the definition of localization to allow for the reporting of multiple sets.}. On observing $X_i$, we must either localize $\theta_i$ by specifying which of $\{1,\dots,L_i\}$ it belongs to, or refrain from making any claim at all about $\theta_i$ (the latter reflecting the decision ``not enough evidence to decide"). 
The corresponding natural notion of error for a given procedure is a {\it false localization rate} (FLR),
\[
\textnormal{FLR} := \EE{\frac{\# \text{false localizations}}{\# \text{localizations made}}}.
\]
As we will see below, the false localization rate generalizes the false discovery rate. 

\begin{definition}[LORD-CI for localization]\label{def:lord:ci:fod} 
Let $\I_i: \X_i \times [0,1] \to 2^{\Theta_i}$
be an arbitrary pre-specified monotone marginal CI rule for $i=1,2,\dots$, and define $\S_i$ as follows: 
$$
S_i = 
\begin{cases}
1, & \text{if $I_i = \I_i(X_i, \alpha_i)$ is a subset of exactly one of $K_{i1},\dots,K_{iL_i}$}\\
0, & \text{otherwise}
\end{cases}.
$$
Then {\it LORD-CI for localization} is the online CI protocol that applies LORD-CI to the above selection rule, and when $S_i=1$, it outputs the unique index $j \in \{1,\dots,L_i\}$ such that $I_i \subseteq K_{ij}$.
\end{definition}

\noindent The above procedure comes with the following guarantee.

\begin{theorem}
%Let $\I_i,\ i=1,2,\dots$ be a sequence of arbitrary marginal CIs. 
The LORD-CI for localization procedure (Definition \ref{def:lord:ci:fod}) satisfies 
% that $\mathrm{mFLR}(T)\leq \alpha$ for any $T$.
% If $\I_i$ are monotone rules, then Procedure \ref{def:lord:ci:fod} also enjoys 
$\textnormal{FLR}(T)\leq \alpha$ for any $T$. 
\end{theorem}

\begin{proof}
Note that the selection rule in Definition \ref{def:lord:ci:fod} can be rewritten as
\begin{equation}\label{eq:lord:ci:fod:sel}
\S_i(X_i,I_i) = 1 \iff X_i \in \{x: \I_i(x,\alpha_i) \subseteq K_{il} \text{ for some $l$} \},
\end{equation}
which defines a predictable selection rule because $\alpha_i$ are predictable. 
Thus, the procedure in Definition \ref{def:lord:ci:fod} is LORD-CI for a predictable selection rule. 
% Therefore, we have $\mFCR(T)\leq \alpha$ for any $T$. 
% For the second part of the theorem, suppose that $\I_i$ are monotone. 
Because the CI rules $\I_i$ are monotone, and the $\alpha_i$ output by the LORD-CI algorithm are also monotone by construction, we conclude that the {\it selection rule} \eqref{eq:lord:ci:fod:sel} is also monotone according to condition \eqref{eq:monotone}. 
Hence, the procedure in Definition \ref{def:lord:ci:fod} is now the LORD-CI procedure for a predictable {\it and monotone} selection rule, which controls the FCR by Theorem \ref{thm:lord:ci:monotone}. 
The last step is to observe that a false localization event implies a false coverage event (but not necessarily the other way around), and hence $\textnormal{FLR}(T) \leq \FCR(T) \leq \alpha$.
\end{proof}

Next we consider some special cases of localization and their implications. 
% Again, we do not view this as a restrictive. 
%Next, we show that with an appropriate configuration of this instantiation of LORD-CI, we can recover online FDR procedures and extend them. 
%tying selection to the confidence interval procedure, can lead to LORD-CI procedures of practical interest. 
%In particular, we will show how {\it a specific configuration} of LORD-CI generalizes the LORD (or LORD++) online testing procedure, and how it implies directional-FDR control when supplementing rejections with natural directional decisions. 

\subsection{Online composite hypothesis testing with FDR control}\label{subsec:FDR}
Suppose that we have a sequence of composite null hypotheses that we wish to test:
$$
H_{0}^i: \theta_i \in \Theta_{0i},\ \ \ \ \ \ \ i=1,2,\dots,
$$
where $\Theta_{0i}\subseteq \Theta_i$. 
For any online testing procedure, let $R_i := \mathbbm{1}(H_0^i \text{ is rejected})$ and define 
$$
\FDR(T) := \EE{\frac{ \#\{i\leq T: R_i = 1, \theta_i\in \Theta_{0i}\} }{ \#\{i\leq T: R_i = 1\} }},
$$
which reduces to the usual definition of the FDR when $\Theta_i$ include a single value, i.e., when testing point null hypotheses. 
%Existing online FDR protocols require computing p-values, which in the composite case is not always trivial to do.\footnote{For example, consider testing $H_0^i: (\theta_{i2}-\theta_{i1})/\theta_{i1}\leq 0.1$ where $\theta_i = (\theta_{i1}, \theta_{i2})^T$ and $X_i\sim N_2(\theta_i, I)$ for $i\geq 1$.}
We can use the procedure of Definition \ref{def:lord:ci:fod} to devise an online {\it testing} protocol that controls the FDR. 
% Specifically, we use $L_i=1$ and $K_{i1} = \Theta_{0i}^c$, and we set $R_i = S_i$ for all $i$ (we reject the null if a selection occurs). In other words, we reject the $i$th null if and only if $\I_i(X_i,\alpha_i) \cap \Theta_{0i} = \emptyset$.
\begin{definition}[LORD-CI for composite testing] \label{proc:comp-test} 
Consider an arbitrary marginal CI rule $\I_i$ for each parameter $\theta_i$. 
We reject the $i$th composite null hypothesis and set $R_i = 1$, if and only if 
\begin{equation}
\label{eq:composite}
\I_i(X_i,\alpha_i) \cap \Theta_{0i} = \emptyset,
\end{equation}
where $\alpha_i$ is determined by the LORD-CI procedure using $S_i=R_i$.
\end{definition}
\noindent The procedure in the definition above comes with the following guarantee. 

\begin{corollary}
The LORD-CI procedure for composite testing (Definition \ref{proc:comp-test}) enjoys $\FDR(T)\leq \alpha$ for any $T$. 
\end{corollary}

\begin{proof}
Specialize the prescription in Definition \ref{def:lord:ci:fod} by taking $L_i=1$ and $K_i = \Theta_i \setminus \Theta_{0i}$. 
Then, $S_i = 1$ if and only if condition \eqref{eq:composite} holds, meaning that $R_i=S_i$ for all $i$.  The use of LORD-CI guarantees that we have $\FCR(T)\leq q$ for any $T$. 
Last, we have that $\FDR(T)\leq \FCR(T)$ simply because a false discovery implies necessarily that a non-covering CI was constructed. 
\end{proof}

%: suppose that we have an arbitrary marginal CI rule $\I_i$ for each parameter $\theta_i$, and specialize the prescription in Definition \ref{def:lord:ci:fod} by taking $L=1$ and $K_i = \Theta_i \setminus \Theta_{0i}$. 
%Thus, $\I_i(X_i,\alpha_i)$ is reported ($S_i = 1$) by the resulting online CI protocol if and only if 
%$$
%\I_i(X_i,\alpha_i) \cap \Theta_{0i} = \emptyset. 
%$$
%Then we claim that setting $R_i = S_i$ for all $i$, controls the FDR. 
%Indeed, the online-CI procedure just defined is a special case of the procedure in Definition \ref{def:lord:ci:fod}, and the CI rules were assumed to be monotone. 
%Meanwhile, we have that $\FDR(T)\leq \FCR(T)$ simply because a false positive implies necessarily that a non-covering CI was constructed. 
%

Before proceeding, we would like to point out a connection to existing online FDR testing protocols. 
We can define a p-value for testing $H_{0}^i$ by 
$$
P_i := \sup\{\alpha: \I_i(X_i,\alpha) \cap \Theta_{0i}\neq \emptyset \},
$$
where $\I_i$ is any monotone CI for $\theta_i$. 
Indeed, if $\theta_i\in \Theta_{0i}$, then for any $t\in [0,1]$ we have
$$
\textnormal{Pr}_{\theta_i}\!\left\{ P_i\leq t \right\} \leq \textnormal{Pr}_{\theta_i}\!\left\{ \I_i(X_i,t) \cap \Theta_{0i} = \emptyset \right\} \leq \textnormal{Pr}_{\theta_i}\!\left\{ \theta_i \notin \I_i(X_i,t) \right\} \leq t.
$$
We can therefore apply an existing online FDR protocol using this definition for a p-value. 
Note that while the computation of the p-value above might not be trivial, we are really only required at each step to check if $P_i\leq \alpha_i$, which is equivalent to rejecting when $\I_i(X_i,t) \cap \Theta_{0i} = \emptyset$. 
In fact, if we use the same CI rules and the same algorithm to set the $\alpha_i$ as in the CI procedure employed in Definition \ref{proc:comp-test}, we obtain exactly the composite testing procedure of Definition \ref{proc:comp-test}.

\subsection{Online sign-classification with FSR control}\label{subsec:FSR}
%\arcomment{Cite Matthew Stephens, Sanat Sarkar and other papers in this section or elsewhere (ambivalent to location of references, but need more.}
Sometimes we would like to ask about the {\it direction} of the effect rather than test a two-sided null hypothesis. 
As argued in \citet{gelman2012we, gelman2000type}, this often makes a more sensible question than asking whether a parameter is {\it equal} to zero. In fact, even statisticians that use a two-sided test of a point null hypothesis, tend to supplement---perhaps with a leap of faith---a rejection of the null with a claim about the sign of the parameter \citep[][call this {\it post hoc} inference of the sign]{goeman2010three}.
Inferring the signs of multiple parameters simultaneously was considered at least as early as \citet{bohrer1979multiple, bohrer1980optimal, hochberg1986multiple}. 
In the story of Section~\ref{sec:intro}, the management might be interested primarily in identifying which drugs have a positive treatment effect and which drugs have a nonpositive treatment effect. 
Throughout this subsection suppose that $\Theta_i\equiv \Theta \subseteq \R$, $\X_i\equiv \X$ and that $X_i\sim f(x_i;\theta_i)$ for some common likelihood function $f:  \X \times \Theta \to [0,\infty)$, so that a common CI rule can be used at all times\footnote{Note that a CI rule depends on the likelihood function only, not on the true value of $\theta_i$}; for lack of a better phrase, we call this situation the ``common likelihood'' case. 

%Sometimes having a testing procedure that controls the FDR is not enough. 
%Indeed, in almost any realistic situation, a point null hypothesis is never exactly true, in which case the FDR of any procedure is trivially zero (simply because there is technically no option to commit a type I error). 
%\citet{tukey1991philosophy} argued against the focus of statisticians on the testing problem:
%
%\begin{quote}
%statisticians classically asked the wrong question---and were willing to answer with a lie \dotsThey asked ``Are the effects of A and B different?" and they were willing to answer ``no." 
%All we know about the world teaches us that the effects of A and B are always different---in some decimal place---for any A and B \dotsWhat we should be answering first is ``Can we tell the direction in which the effects of A differ from the effects of B?" In other words, can we be confident about the direction from A to B? Is it ``up", ``down" or ``uncertain"? \dotsThe follow-up question is about how much---about what we are confident of concerning the numerical difference 
%\[
%\text{effect of A MINUS effect of B}.
%\]
%\end{quote}
%
%As discerningly noted by \citet{gelman2000type}, the problem is that we tend to casually replace rejections of hypotheses with statements about the signs of the parameters, while the error rate for the sign might be substantially higher than the Type-I error rate. 

When considering a sign-classification procedure, we will aim to control---in analogy to the FDR---the expected ratio of number of incorrect directional decisions to the total number of directional decisions made. 
Throughout this paper, to make a directional decision means to classify $\theta_i>0$ (positive) or $\theta_i\leq 0$ (non-positive); because zero is included on one side, this can be considered a {\it weak} sign-classification (although the definition is not symmetric, zero can be just as well appended to the positive side instead of the negative side). 
Hence, a sign-classification protocol is an online procedure that outputs
$$
D_i = 
\begin{cases}
1,& \text{if $\theta_i$ classified as positive}\\
-1,& \text{if $\theta_i$ classified as non-positive}\\
0,& \text{if no decision on the sign of $\theta_i$ is made}
\end{cases}.
$$
Borrowing a term from \citet{stephens2016false}, we define the {\it false sign rate} as 
\[
\FSR(T) := \EE{ \frac{ \#\{i\leq T: \theta_i\leq 0, D_i=1\} + \#\{i\leq T: \theta_i>0, D_i=-1\} }{\#\{i\leq T: D_i=1\} + \#\{i\leq T: D_i=-1\}} }.
%\begin{aligned}
%\FSR(T) &= \EE{ \frac{\#\text{nonpositive parameters classified as positive or positive parameters classified as nonpositive}}{\#\text{parameters whose sign was classified}} }\\
%&= \EE{ \frac{ \#\{i\leq T: \theta_i\leq 0, D_i=1\} + \#\{i\leq T: \theta_i>0, D_i=-1\} }{\#\{i\leq T: D_i=1\} + \#\{i\leq T: D_i=-1\}} }
%\end{aligned}
\]
It is worth noting that this is slightly different from the definitions of \citet{benjamini1993false}, who consider procedures that classify parameters as strictly positive or as strictly negative; however, if there are no parameters that equal zero exactly, then all definitions coincide \citep[which is the case in virtually all realistic situations, see, e.g.,][]{tukey1991philosophy}. 
A natural procedure to consider is applying any online FDR protocol to test the hypotheses that $\theta_i=0$, and then classify each rejection according to the sign of an unbiased estimate of $\theta_i$ (so, for example, when $\theta_i = \EE{X_i}$, a rejected null with $X_i>0$ entails $D_i = 1$). 
We will see later that, for example, applying LORD++ to the usual two-sided $p$-values indeed works, however FSR control is not {\it automatically} guaranteed, i.e., this requires a proof \citep[see][who point out caveats in replacing rejections with statements about the signs of the parameters]{gelman2000type}. 
%Moreover, we will see later that using two-sided $p$-values (with an online FDR procedure) is anyway not optimal here in terms of power. 

\smallskip
We can rely again on the procedure of Definition \ref{def:lord:ci:fod} to devise a sign-classification protocol that controls the FSR. 
Thus, suppose that we have an arbitrary (common) marginal CI rule $\I(\cdot,\cdot)$. 
Now specialize the prescription in Definition \ref{def:lord:ci:fod} by taking $L_i\equiv 2$, $K_{i1}\equiv (-\infty,0], K_{i2} \equiv (0,\infty)$. 
In words, this is the LORD-CI procedure that reports $I_i=\I(X_i,\alpha_i)$ whenever it includes either only positive or only non-positive values. 
This special case of the procedure in Definition \ref{def:lord:ci:fod} is central enough to merit a separate definition. 

\begin{definition}[Sign-determining LORD-CI procedure]\label{def:sd:lord:ci}
Suppose that we are in the ``common likelihood" case, and let $\I: \X \times [0,1] \to 2^{\Theta}$ be any marginal confidence interval procedure, i.e., $\textnormal{Pr}_{\theta_i}{\left\{\theta_i\notin \I(X_i,a)\right\}}\leq a$ for any $a\in [0,1]$. 
Assume that for any parameter $\theta_i$ there is a corresponding ``null" value $\theta_{0i}\in \Theta$. 
The {\it sign-determining} LORD-CI procedure associated with $\I$ is defined to be the LORD-CI procedure that utilizes the  selection rules 
\begin{equation}\label{eq:sd-rule}
S_i = 
\begin{cases}
1, \ &\text{if } \{\tau - \theta_{0i}: \tau\in \I(x,\alpha_i)\}\subseteq (0,\infty)\ \ \ \textnormal{or}\ \ \ \{\tau - \theta_{0i}: \tau\in \I(x,\alpha_i)\}\subseteq (-\infty,0]\\
0, \ &\textnormal{otherwise}
\end{cases},
\end{equation}
and constructs $I_i = \I(X_i,\alpha_i)$ if $S_i=1$. 
\end{definition}
For simplicity, assume from now on that $\theta_{0i}\equiv 0$. 
In that case the sign-determining LORD-CI procedure constructs $I_i = \I(X_i,\alpha_i)$ if and only if this interval is  {\it sign-determining}, meaning that it includes only positive or only nonpositive values. 

\smallskip
Returning to the FSR problem, apply now the CI procedure from Definition \ref{def:sd:lord:ci} with an arbitrary choice of $\I$, and set 
\begin{equation}
\label{eq:sign:classification}
D_i = 
\begin{cases}
1,& \text{if $S_i = 1$ and $I_i \subseteq (0,\infty)$}\\
-1,& \text{if $S_i = 1$ and $I_i \subseteq (-\infty,0]$}\\
0,& \text{if $S_i = 0$}
\end{cases}.
\end{equation}
Then we have the following result:

\begin{corollary}
The sign-classification procedure given by \eqref{eq:sign:classification}, enjoys $\FSR(T)\leq \alpha$ for any $T$. 
\end{corollary}

\begin{proof}
We have $\FSR(T)\leq \FCR(T)$ because a wrong decision on the sign of a parameter necessarily implies that a non-covering CI was constructed. 
Here, the left hand side of the inequality is the false sign rate associated with the sign-classification procedure in \eqref{eq:sign:classification}, and the right hand side is the false coverage rate associated with the sign-determining LORD-CI procedure of Definition \ref{def:sd:lord:ci} (using $\theta_{0i}\equiv 0$). 
On the other hand, $\FCR(T)$ is controlled as a special case of the procedure of Definition \ref{def:lord:ci:fod}, because we assumed that $\I$ is monotone. 
\end{proof}

%It is also straightforward to see that the framework developed above can easily be applied to the framework of three-sided hypothesis testing as developed by \cite{goeman2010three}.
%\awcomment{Is what we're doing any different from three-sided testing?}

\begin{remark}
It is easy to see that we could drop the assumption on monotonicity of the CI rules in this section and still be guaranteed control of the respective modified error rates, for example $\mFDR$ in Subsection \ref{subsec:FDR} and $\mFSR$ in Subsection \ref{subsec:FSR} (which would now be implied by $\mFCR$ control for the corresponding CI procedure). 
\end{remark}

\subsection{Configuring the sign-determining LORD-CI procedure}\label{subsec:sd:lord:ci}
The sign-determining LORD-CI procedure was used in the previous subsection as a ``wrapper'' device to control the FSR for any definition of the CI rules $\I_i$, but it can be of interest to design specific $\I_i$ rules since we know that the FSR procedure will only select sign-determining CIs. 
Indeed, in most realistic situations, it is useful to supplement a directional decision with confidence bounds that are consistent with that decision. 
For example, if the team of statisticians declare a specific drug to have a positive effect, the management will likely want to know how large the effect is at least, as would be quantified by a nonnegative lower endpoint of a CI. 

\smallskip
Thus, ideally, the sign-determining LORD-CI procedure selects and constructs a large number of CIs---meaning that it is ``powerful'' when translated to an FSR protocol as in Section \ref{subsec:FSR}---while the lower endpoint for a positive interval is as far away from zero as possible (and, similarly, the upper endpoint for a nonpositive interval is as far away from zero as possible). 
Unfortunately, these two goals are conflicting in general; see \citet{benjamini1998confidence}, who study the single-parameter case. 
The tradeoff between these two properties will be controlled here through the choice of the marginal CI rule $\I_i$; thus, the corresponding sign-determining LORD-CI procedure may be seen as the online counterpart of the offline sign-determining multiple testing procedure of \citet{weinstein2014selective}. 
Below, we point out a few concrete examples of CI rules $\I_i$. 
For the rest of this section assume that $X_i\sim N(\theta_i,1)$, though the constructions below can be extended beyond the normal case. 

%\begin{quote}
%What we should be answering first is ``Can we tell the direction in which the effects of A differ from the effects of B?" In other words, can we be confident about the direction from A to B? Is it ``up", ``down" or ``uncertain"? \dotsThe follow-up question is about how much---about what we are confident of concerning the numerical difference 
%\[
%\text{effect of A MINUS effect of B}.
%\]
%\end{quote}

%\medskip
%More generally, and in accordance with Tukey's ``wish-list", we will follow the pursuit of \citet{weinstein2014selective} to construct {\it sign-determining} CIs, that is, CIs that include either only positive or only nonpositive values, while controlling the FCR. 
%Just as explained in \citet{weinstein2014selective} for the offline setting, any such procedure immediately translates---by classifying $\theta_i$ as positive or nonpositive according as the interval is a subset of $(-\infty,0]$ or $(0,\infty)$---to a sign-classification procedure that controls the directional-FDR. 
%
%By choosing different marginal CI procedures $\I$, we obtain sign-determining LORD-CI procedures with different properties: a sign-determining LORD-CI equipped with a marginal CI that avoids values of opposite signs for relatively early values of $x$, will make more selections, which we refer to as having more power; 
%and, loosely speaking, choosing short marginal CIs will result in short constructed intervals. 

\begin{enumerate}
\item $\I$ is the usual symmetric interval, $\I(x,\alpha) = (x-z_{\alpha/2}, x+z_{\alpha/2})$. 

It can be easily verified that the sign-determining LORD-CI procedure with this choice for $\I$, selects exactly the set of parameters rejected by the LORD++ online FDR procedure using the usual two-sided $p$-values
$$
P_i = 2(1-\Phi(|X_i|)). 
$$
A constructed CI has length $2z_{\alpha_i/2}$, and is guaranteed to be sign-determining. 
As a byproduct, if we translate this into a sign-classification procedure as explained in Section \ref{subsec:FSR}, we have as a conclusion that selecting with LORD++ and classifying according to the sign of $X_i$, controls the FSR. 
In fact, this conclusion still holds if $D_i=-1$ is interpreted as declaring that $\theta_i$ is {\it strictly} negative rather than non-positive, because the usual symmetric interval is open. 

\item $\I$ is the ``one-sided" interval given by 
$$
\I(x,\alpha) = 
\begin{cases}
(x-z_\alpha,x+z_\alpha), & \text{if $0<|x|<z_\alpha$}\\
(0,x+z_\alpha), & \text{if $x>z_\alpha$}\\
(x-z_\alpha,0], & \text{if $x<-z_\alpha$}
\end{cases}.
$$ 
It can be verified that the sign-determining LORD-CI procedure with this choice for $\I$, selects exactly the set of parameters rejected by the LORD++ online FDR procedure using ``one-sided" $p$-values
$$
P_i = 1-\Phi(|X_i|),
$$
hence is much more powerful when translated into a sign-classification procedure. 
However, constructed intervals do not have an a priori bound on their length, since they take the form $(0,X_i+z_{\alpha_i})$ (if the observation is positive) or $(X_i-z_{\alpha_i},0]$ (if the observation is negative). 
Perhaps more seriously, a reported interval necessarily touches zero, thus failing to address our follow-up question on how big the effect is at least (a nonpositive interval even {\it includes} zero).

\item $\I$ is the Modified Quasi-Conventional (MQC) CI of \citet{weinstein2014selective}. 

%Tukey's ``follow-up" question concerning the magnitude of the effect, calls for an actual {\it interval} for $\theta_i$, not just a directional decision. 
%\citet{benjamini1998confidence} interpreted Tukey's ``requests" as asking for an interval that stops including values of opposite signs for relatively small values of the observation, but still has finite endpoints. 
%Using such an interval in our sign-determining LORD-CI procedure, will result in more power to detect the sign than that in the previous item. 
The MQC confidence interval is in a sense a compromise between the two choices of $\I$ presented above: it determines the sign earlier than the two-sided interval but not as early as the ``one-sided" interval.
In turn, it leads to more power than LORD++ applied to two-sided $p$-values when interpreted as a sign-classification procedure, and at the same time separates from zero for large enough $x$. 
A mathematical definition of the MQC interval is given in \citet{weinstein2014selective}, where its properties are further explained; we include a figure instead to illustrate the properties of that interval. 
In Figure \ref{fig:mqc} the endpoints of the MQC interval $\I(x,\alpha)$ are shown in solid lines as a function of the observation $x$ for $\alpha=0.1$.
The potential gain in power due to using the MQC interval instead of the symmetric interval, is demonstrated in Section \ref{sec:experiments}.

\end{enumerate}

\begin{figure}[H]
  \centering
    \includegraphics[width=.65\textwidth]{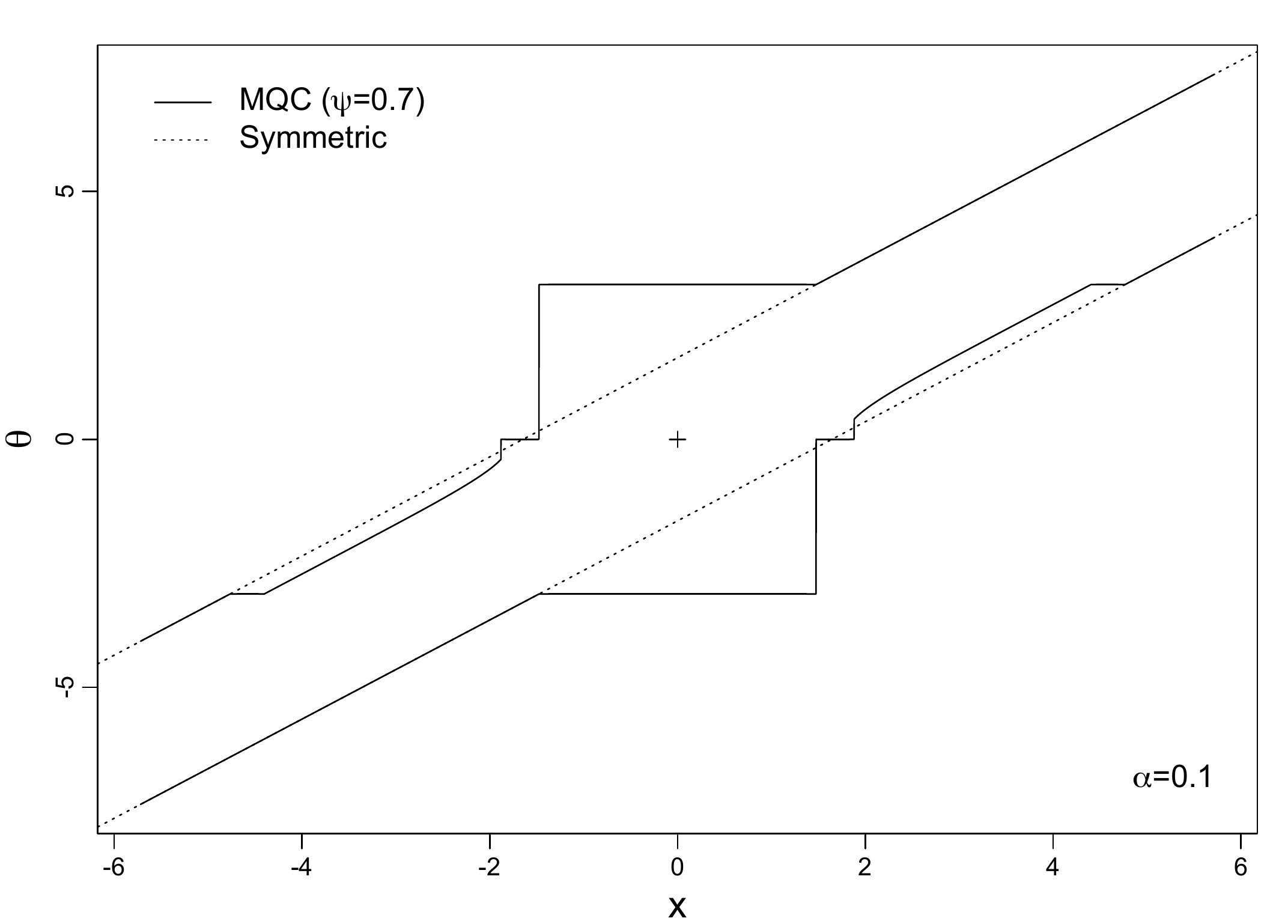}
  \caption{The Modified Quasi-Conventional (MQC) CI rule. 
  % The horizontal axis is the observation value $x$, and the vertical axis is $\theta$ values that would be covered by the MQC rule on observing $x$.  
  For each observation $x$, the two solid lines show the lower and upper endpoints of the corresponding MQC CI that would be constructed for $\theta$. Dotted lines are lower and upper endpoints of the usual two-sided CI. It can be seen that the MQC interval excludes values of opposite signs earlier, that is at a smaller $x$ value, than the usual two-sided CI. (The parameter $\psi\in (0.5,1)$ controls how early sign-determination occurs, and here we used $\psi=0.7$.) The unusual constant shape of the MQC in the neighborhood of zero does not matter because intervals that cross zero are anyway discarded in a sign-determining selective-CI procedure; see the discussion in \citet{weinstein2014selective} which applies also for the online case treated here.}
  \label{fig:mqc}
\end{figure}

\section{Numerical experiments} \label{sec:experiments}

\subsection{Simulations} \label{subsec:simulations}

To examine how the LORD-CI procedure compares to conditional CIs, we carry out numerical experiments where online confidence intervals are constructed under different (predictable) selection schemes. 
We set $\alpha=0.1$ and in each of $N=10,000$ simulation runs, we draw $m=10,000$ parameters i.i.d.~from a mixture
$$
\theta_i=
\begin{cases}
10^{-3},&\text{w.p. }0.45\\
-10^{-3},&\text{w.p. }0.45\\
1 + W_i,&\text{w.p. }0.1\\
\end{cases},
$$
where $W_i\sim \text{Pois}(1)$. 
The mass at $\pm 0.001$ represents the ``null" component (essentially zero), while the ``nonnulls" are drawn so that large effects are rare. 
The observations are then drawn as $X_i \sim N(\theta_i,1)$. 
The $X_i$ are revealed one by one, and a confidence interval is to be quoted whenever a parameter is selected. 
The LORD-CI procedure uses the sequence of $\alpha_i$ specified by the LORD++ procedure \citep{RYWJ17} with ``default" choices $W_0 = \alpha/2$ and $\gamma(j) = 0.0722 \frac{\log(j\vee 2)}{je^{\sqrt{\log j}}}$, as used in the experiments of \cite{javanmard2016online,RYWJ17}. 
{\it If not indicated otherwise}, the marginal CI used for LORD-CI is the symmetric two-sided interval, and the conditional CI used is the construction from \citet[][Section 2]{weinstein2013selection} obtained by inverting shortest acceptance regions. 
Table \ref{tab:simulations} gives quantitative summary statistics (averaged over the $N$ replications) for the three simulation examples. 
Below, we examine the output from a single realization of the experiment for each of the examples. 

\smallskip
We begin with a simple selection rule, where a CI is constructed when $|X_i|>3$, i.e., when the size of the current observation exceeds a fixed threshold. 
Figure \ref{fig:simulation:threshold} shows conditional CIs (red) versus LORD-CI intervals (black) for a single realization. 
The conditional CIs are considerably shorter than LORD-CI, which seems to be conservative with $\FCP = 0.043$ as compared to about 0.1 for conditional. 
In particular, the conditional CIs become closer to the marginal two-sided $90\%$ interval as the observation size increases, which would in that sense resemble Bayesian credible intervals for our example (these are not shown in the figure). 
Both the conditional and LORD-CI intervals may cross zero, as can be seen in the plot. 
In fact, as many as 53\% of the conditional CIs cross zero, and 38\% of LORD-CI intervals cross zero. 
Note that the lower endpoint of the CI is monotone non-decreasing for the conditional intervals, but not for LORD-CI. 
The conditional intervals seem preferable in this situation. 

\begin{figure}
  \centering
    \includegraphics[width=.65\textwidth]{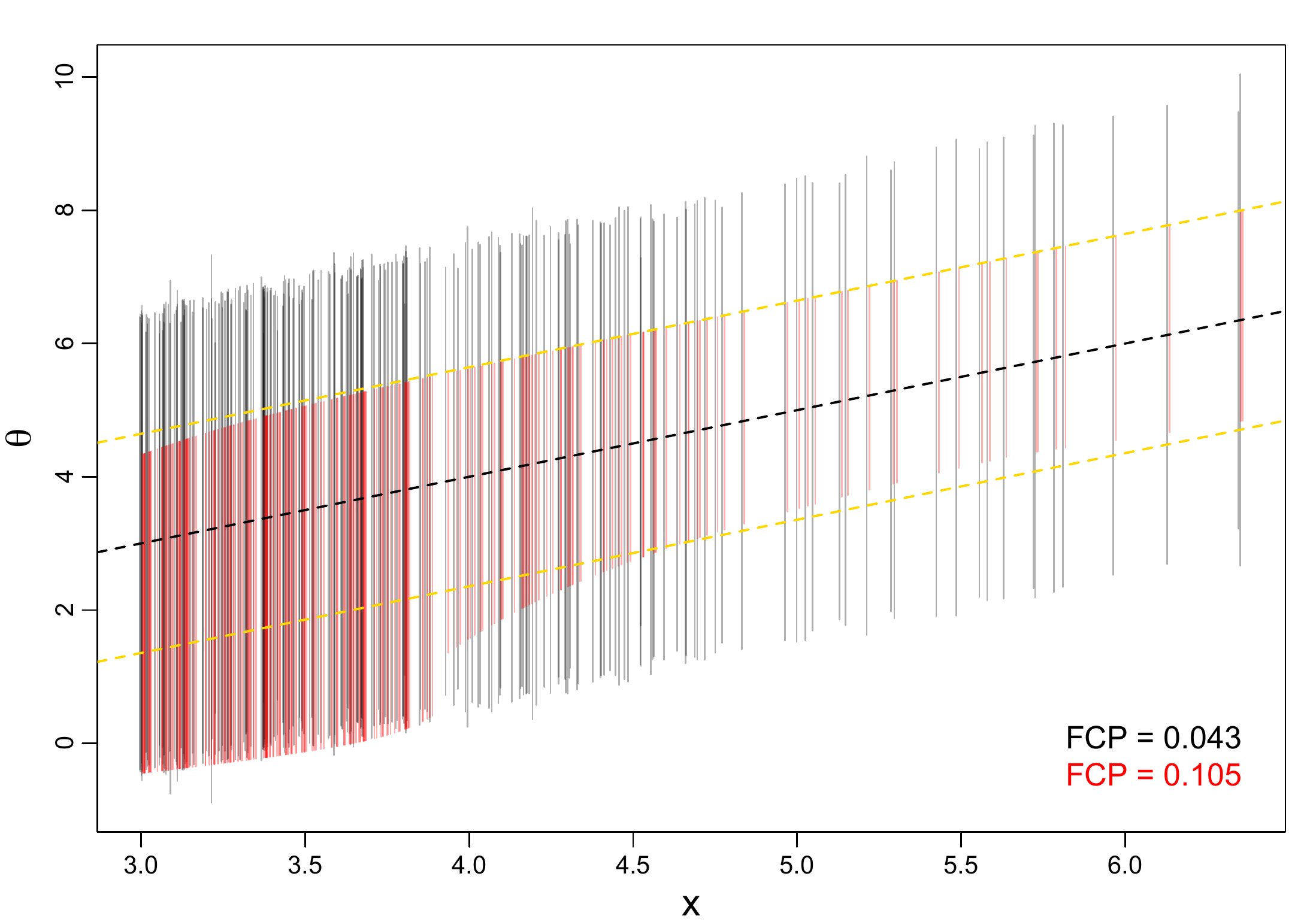}
  \caption{Conditional Vs. LORD-CI intervals for selection above a fixed threshold. 
  The figure shows constructed intervals as vertical segments: conditional CIs in red, LORD-CI in black. 
  Black dashed line is the identity line, and the yellow dashed lines represent the endpoints of the marginal (unadjusted) $90\%$ interval. 
  }
  \label{fig:simulation:threshold}
\end{figure}

\smallskip
The second simulation example illustrates a situation where we are interested first in detecting the sign of the parameters, and second in supplementing a directional decision with confidence bounds. 
For this we implement the sign-determining LORD-CI procedure of Section \ref{sec:localization}, in other words, $\theta_i$ is selected whenever the candidate (symmetric) LORD-CI interval excludes zero. 
Because this amounts to selecting $\theta_i$ when $|X_i|>\Phi^{-1}(1-\alpha_i/2)$, and because $X_i$ are independent and $\alpha_i$ predictable, the conditional distribution of $X_i$ given $S_i=1$ and $S_1,\dots,S_{i-1}$, is that of a truncated normal, and we can use again the intervals of \citet{weinstein2013selection} (the cutoff will now be different for every selection, as opposed to the previous example where it was always at $3$). 
As the conditional intervals converge to the marginal two-sided CI for observation values much larger than the threshold, they are significantly shorter than LORD-CI for these instances. 
However, the conditional intervals may cross zero, whereas the LORD-CI intervals are (by design) guaranteed to be sign-determining. 
In fact, more than half (51.2\%) of the conditional CIs include both positive and negative values; as mentioned in the introduction, in many applications it would be desirable to report intervals that do not cross zero, and in that case LORD-CI has a clear advantage. 
In practice, one might be tempted to ignore the conditional CIs that include zero; of course, with this strategy we lose FCR control---in our example, the FCP would increase to $0.2$ if we kept only the intervals that do not include zero (see Section \ref{subsec:inconsistency} below for further discussion). 

\begin{figure}
  \centering
    \includegraphics[width=.85\textwidth]{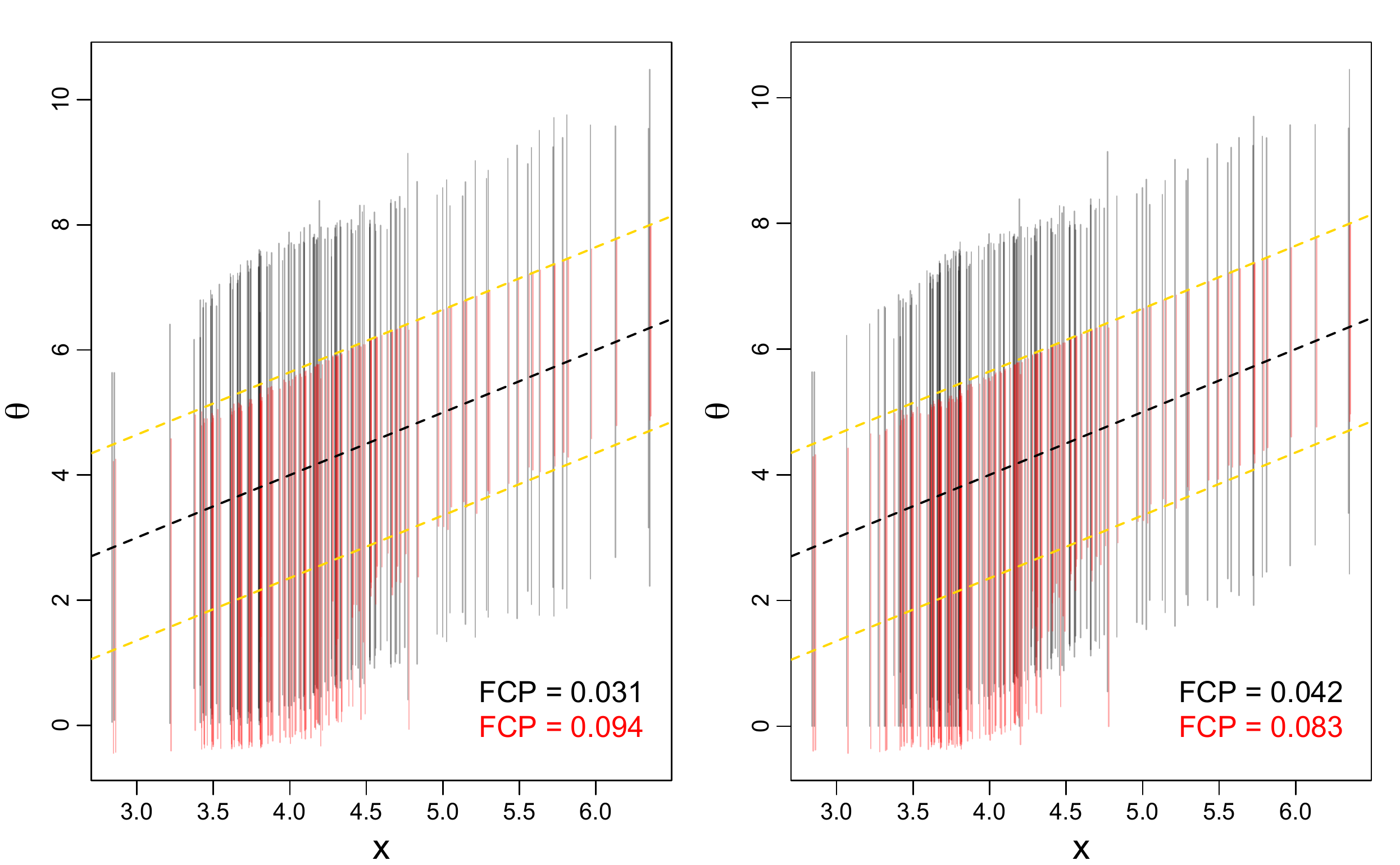}
  \caption{Conditional Vs. LORD-CI intervals for selection by the sign-determining LORD-CI procedure. 
  The figure shows constructed intervals as vertical segments: conditional CIs in red, LORD-CI in black. 
  {\it Left}: selection utilizes the usual symmetric CI. 
  {\it Right}: selection utilizes the MQC interval. 
  LORD-CI intervals are sign-determining (by design), while conditional CIs may cross zero. 
  More parameters are selected in the right panel (higher power). 
  Black dashed line is the identity line, and the yellow dashed lines represent the endpoints of the marginal (unadjusted) $90\%$ interval. 
  }
  \label{fig:simulation:lordci}
\end{figure}

\smallskip
In the next simulation we demonstrate the advantages of equipping the sign-determining LORD-CI procedure with a marginal CI that itself has early sign determination, as discussed in Section \ref{sec:localization}. 
Thus, we consider the same realization of the data as in the previous example, but now parameters are selected through the sign-determining LORD-CI procedure using the MQC interval of Section \ref{sec:localization} instead of the usual symmetric CI. 
This resulted in 144 selections, about 13\% more than the number of selections with the symmetric CI; furthermore, the set of parameters selected by the symmetric CI-equipped procedure is (always) a subset of the parameters selected by the MQC-equipped procedure. 
The FCP for the (sign-determining) LORD-CI intervals is still quite low (0.042) relative to the nominal level. 
When constructing conditional CIs for the selected parameters, again 50\% of the intervals cross zero, while the LORD-CI intervals are all sign-determining. 
Figure \ref{fig:simulation:lordci} displays adjusted (LORD-CI) marginal CIs along with conditional CIs for parameters selected by the sign-determining LORD-CI procedure using the symmetric CI (left panel) and the MQC CI (right panel). 
Figure \ref{fig:simulation:lordci:bytime} is a representation by time of the constructed LORD-CI intervals: top panel for the procedure using the symmetric CI, and bottom panel for that using the MQC CI. 
Each black circle represents a data point $X_i$, and a vertical segment is shown whenever selection occurs---green for a covering CI and red for a non-covering CI. 
There are 127 constructed CIs in the top panel and 144 in the bottom panel.

\begin{figure}[]
  \centering
    \includegraphics[width=.825\textwidth]{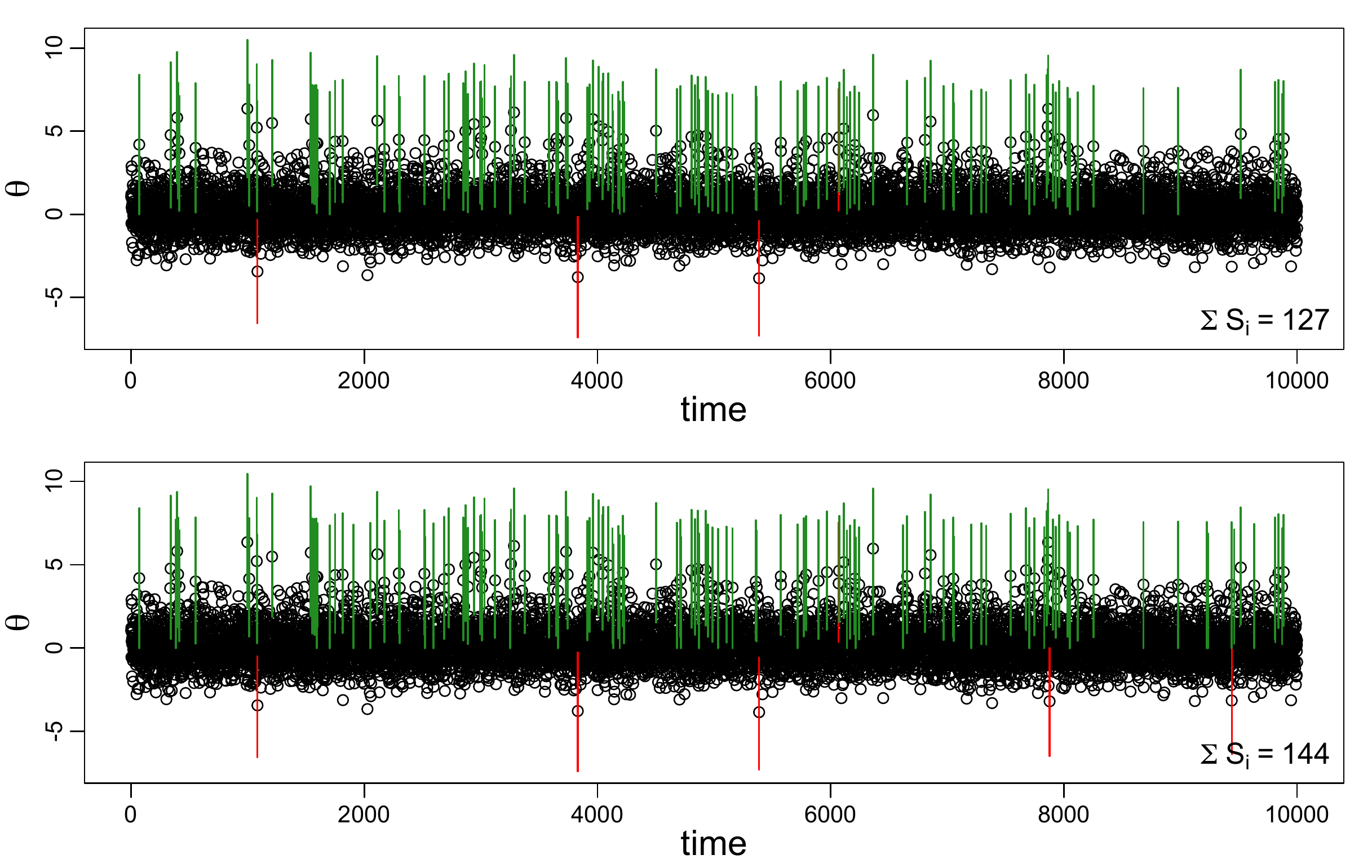}
  \caption{LORD-CI intervals for sign-determining LORD-CI procedure, representation by time. Top panel is for selection with the symmetric CI, bottom is for MQC. The horizontal axis is ``time" (order of appearance) and the vertical position of the circles is the $X_i$ values. When a CI is not constructed, only a circle is shown; for selected parameters, a CI is shown---green for a covering CI and red for a noncovering CI. The MQC-equipped procedure makes 17 more selections (an increase of 13\%). 
  }
  \label{fig:simulation:lordci:bytime}
\end{figure}

\begin{table}[]
    \begin{center}
\begin{tabular}{| *{9}{c|} }
    \hline
    & \multicolumn{2}{c|}{Fixed threshold}
            & \multicolumn{2}{c|}{Sgn-det LORD-CI (Symm)}
                            & \multicolumn{2}{c|}{Sgn-det LORD-CI (MQC)}                \\
    \hline \hline
   &   LORD-CI  &   cond.  &   LORD-CI  &   cond.  &   LORD-CI  &   cond.  \\ 
    \hline
\FCR   &  0.028   &   0.1  &   0.03  &  0.1   &   0.032   &  0.1   \\
    \hline
\mFCR   &  0.028   &   0.1  &  0.031   &  0.1   &   0.032   &  0.1   \\
    \hline
$\EE{\sum_i S_i}$   &  253.396   &  253.396   &  133.49   &  133.49   &   154.395   &   154.395  \\
    \hline
$\EE{(\sum_i S_i \cdot A_i)/\sum_i S_i}$   & 0.649 & 0.502   &  1   &  {0.533}   &  1   &  0.527  \\
\hline
\end{tabular}
    \end{center}
    \caption{Simulation summary for three selection schemes. 
    In the last row, $A_i$ is the indicator for the event that the CI is sign-determining. 
    The main and important advantage of using LORD-CI here, is that it is possible to guarantee that only sign-determining intervals are constructed. 
    A disadvantage of LORD-CI is that the FCR seems to be significantly smaller than the nominal level (compare to FCR for conditional), and intervals for ``large" observations seem to be excessively long.} 
%    \color{red}{explain takeaway messages very briefly. what should they see beyond some numbers? what numbers should they compare to reach what conclusion?}

    \label{tab:simulations}
\end{table}

\subsection{An inconsistency of conditional CIs} \label{subsec:inconsistency}

%\arcomment{I don't like the part of the title that says ``multiple comparisons procedures'' -- I think that phrase, when Tukey coined it, was really for performing many pairwise comparisons between treatments, that is comparing A vs B, A vs C, B vs C, etc. It's not the right phrase for our setting, not sure what is.}
In realistic situations where our (online) model might be applicable, it is almost always the case that the researcher has in mind a question of primary importance and one (or more) of secondary importance. 
In the motivating example from the Introduction, the management might be interested first in knowing the sign of the parameters $\theta_i$ (say positive or nonpositive), but would also like to supplement with confidence limits each parameter whose sign was classified. 
In general, it is common practice to answer the question of primary interest by running a multiple comparisons procedure, for example a multiple hypothesis testing rule or, as would apply to our example, a multiple sign-classification rule. 
Because the follow-up question is posed only if the first question was answered (e.g., we want a CI only if we were able to classify the sign), the conditional approach might appear as natural to use at the second stage. 
Nevertheless, the purpose of this section is to demonstrate that constructing conditional CIs after running a multiple comparisons procedure might lead to contradictions. 
Moreover, if one insists on conditional CIs, the price of ``resolving" these incompatibilities might be a serious loss in power. 

\smallskip
Let us return to the motivating story of the Introduction, which we will now accompany with a simulation for illustration. 
Thus, we set $\alpha=0.1$ and draw $m=10,000$ parameters independently such that $\theta_i=(-1)^i\cdot 0.001$ (effectively ``null") with probability 0.8, and $\theta_i = 2$ with probability 0.2. 
These represent the ground truth for the treatment effects of the first $m$ drugs. 
The observations, which we assume arrive independently one at a time, are $X_i\sim N(\theta_i,1)$. 
In the Introduction a CI was reported once $X_i$ exceeded a fixed threshold. 
Suppose now that the statisticians are interested first in classifying the sign of a parameter as positive (``treatment effective") or nonpositive (``treatment ineffective"), and then follow up with CIs for those parameters whose sign was classified. 
To answer the first question, and being aware of multiplicity issues, the team decides to run the LORD++ testing procedure on two-sided $p$-values, where for each rejection they classify the sign as positive or nonpositive according as $X_i$ is positive or nonpositive. 
%That is, an interval will be reported for $\theta_i$ whenever the null $H_{0i}: \theta_i=0$ is rejected by LORD++ operating on the usual two-sided $p$-values. 
This resulted in $76$ selections in our simulation run, and makes sense as a criterion for whether to report an interval or not, because we know from the results of the current paper that the FSR  is controlled. 
Furthermore, remember that, as shown in Subsection \ref{subsec:sd:lord:ci}, constructing the LORD-CI symmetric interval for each selected parameter, ensures at the same time control of the FCR and that none of the constructed CIs include values of opposite signs. 
This is an output the management will, arguably, be content with seeing, at least in the sense that each reported CI is conclusive about the direction of the effect of the corresponding drug  (because the intervals do not cross zero). 

\smallskip
Instead, suppose that the statisticians will actually construct a $90\%$ conditional CI for each selected parameter. 
Now, because we use conditional CIs, it is impossible to ensure that a constructed interval includes values of only one sign (that is, does not cross zero)---this is true no matter what choice we make for the conditional CI rule. 
Here we used the conditional CI of \citet{weinstein2013selection} which inverts shortest acceptance regions. 
The left panel of Figure \ref{fig:appendix} shows the 76 constructed $90\%$ conditional CIs. 
We know that this strategy controls the FCR (in our single realization of the experiment $14.5\%$ of the constructed intervals are non-covering), but, less conveniently, there are also 52 (about $68\%$!) of these that cross zero. 
Hence, the team of statisticians will first have to reconcile the fact that on the one hand, each selected parameter can be safely classified for sign (as far as FSR is controlled), and on the other hand some intervals still include both positive and negative values. 
In any case---even if this incompatibility is overlooked---the management should certainly complain about the CIs that cross zero (because these are ambiguous about the direction of the effect of the corresponding drug). 
Trying to rectify the situation, they may ask to remove all CIs that do cross zero; unfortunately, doing this they would generally lose FCR control. 
The middle panel of Figure \ref{fig:appendix} shows the subset of (original) conditional CIs which do not cross zero; almost half of these ($45.8\%$) do not cover their parameter. 
Nevertheless, the statisticians might propose at this stage to still keep only the CIs that do not cross zero, but re-adjust them for the fact that further selection took place, by constructing again conditional CIs with an appropriate cutoff. 
This will admittedly restore FCR control: the right panel of Figure \ref{fig:appendix} shows the re-adjusted CIs with dashed lines, and the proportion of such intervals that fail to cover their parameter drops again to 0.125. 
The problem is that some of the re-adjusted CI cross zero again (in fact, a much higher proportion than that in the first place), taking us back to the previous stage. 
If we were now to repeat the process by discarding the new $16$ re-adjusted CIs that cross zero, we would be left with only $8$ selections before even adjusting the CIs again. 
In other words, we are already down from the 76 sign-determining LORC-CI intervals to no more than 8 if we use conditional CIs. 
In general, this cycle could continue until there are very few parameters to report a CI for (maybe none). 
We should remark at this point that using a conditional CI that has better sign-determining properties, like the two options in \citet{weinstein2013selection}, could improve the results for the conditional approach, that is, we might end up with more reported CIs. 
However, as we remarked before, the phenomenon in its essence remains regardless of the choice of the conditional CI.

\begin{figure}[h!]
  \centering
    \includegraphics[width=\textwidth]{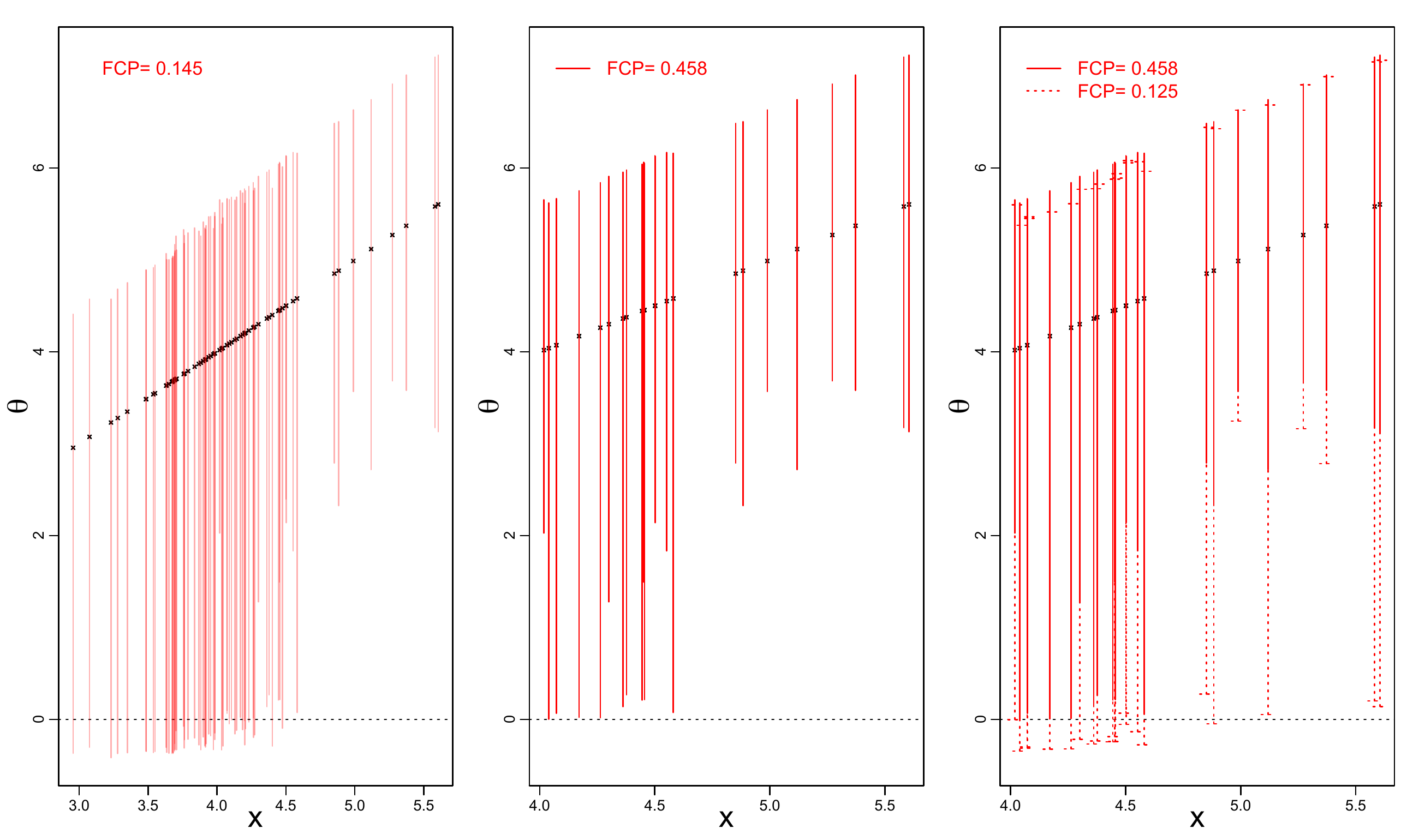}
  \caption{Conditional CIs following LORD++ selection.
  The left panel shows the 76 conditional $90\%$ CIs originally constructed for rejected nulls. 
  In the middle panel we kept only the 24 intervals that do not cover zero. 
  The right panel shows these 24 intervals again (solid lines), but now along with their re-adjusted version (dashed lines). 
  $2/3$ of the re-adjusted conditional CIs again cross zero.} 
  \label{fig:appendix}
\end{figure}

\section{Discussion}\label{sec:discussion}
In a multiplicity problem, there is a sequence of unknown parameters $\theta_i, \ i\leq m\in \N\cup \{\infty\}$, considered simultaneously, and for each we observe data $X_i$ that is informative about $\theta_i$. 
The main task is usually to use the random observations to localize ``interesting" parameters in the sequence, while controlling in some specified sense the rate of erroneous localizations. 
When the multiplicity problem arises in realistic situations, it is often the case that:
\begin{enumerate}
\item The $X_i$ arrive online, and we have to make decisions on the fly. 
\item We are concerned with more complicated problems than testing a sequence of point null hypotheses. 
For example, if $\theta_i$ is the percent difference of mean treatment effects between the $i$-th drug and a market standard, then we might want to know which $\theta_i$ are bigger than 10 and which are smaller than $-5$, as opposed to simply testing whether $\theta_i>0$.
\item There is a follow-up question to the primary one, e.g., if a drug can be declared to improve by at least 10\%, it would be nice to know also how large the improvement is at most.
% \item The $X_i$ are not independent. \arcomment{I disagree that this is a feature of multiplicity problems in realistic online situations. Eg: ATE estimation in clinical trials.}
\end{enumerate}
% Except for the last item, which we only treat very partially in Subsection \ref{sec:asynchronous}, 
We offer methodology to address all of the concerns mentioned in the list above. 
The basic tool is an online algorithm for adjusting the levels of {\it selective} marginal confidence intervals so that we have control over the false coverage rate for any predictable selection rule. 
We show how to instantiate our online CI procedure to provide answers to items 2 and 3 in the list above. 
Specifically, adapted to the example above, we would offer an online procedure that reports only CIs that are either contained in $(10,\infty)$ or contained in $(-\infty,-5)$, while ensuring that $\FCR\leq \alpha$ at any point in time. 
The power of the online FCR procedure as a localization rule, and the shape/length of the reported CIs, is determined by the choice of the underlying marginal CI rule being used. The CI rules are allowed to be arbitrary in this paper, and we offered some insights on their interplay with FSR control in Section~\ref{subsec:FSR} and Section~\ref{subsec:simulations}. Of course, while the real-valued case is of particular interest throughout the paper, the methodology itself applies to any observation space $\mathcal{X}$ and any parameter space $\Theta$. 
% We explored these in the simulations, but several theoretical and practical questions remain open.  
% This is a separate question, and not treated in the current article; we only leverage existing results of \citet{weinstein2014selective} to configure our procedure in the context of the online FSR problem, and obtain what can be considered an online counterpart to their batch procedure. 

\smallskip
We note that other than LORD++, we do not know if it is possible to find FCR analogs of other online FDR procedures like generalized alpha-investing by \cite{aharoni2014generalized}. We conjecture that existing online FDR procedures do not have natural FCR analogs, especially not the recent adaptive algorithms like SAFFRON \citep{ramdas2018saffron}. 
Specifically, other than LORD-CI, we were not able to construct any other online FCR procedures. 

\medskip
We end this article with an important remark. Keeping with historical treatment of multiple testing problems, it felt natural to write this article in terms of parameters and confidence intervals, but our entire methodology for FCR control goes through seamlessly for \emph{prediction intervals} as well. For instance, consider a regression problem where we are given a training dataset $(X^{\text{tr}}_j,Y^{\text{tr}}_j) \sim P_X \times P_{Y|X}$ and a sequence of test points $X_i \sim P_X$ at which we may want to make predictions. We may treat the unknown $Y_i$ in the same way as we treated the unknown $\theta_i$ in this paper. On observing $X_i$, we may deicde whether we wish to report a prediction interval $I_i$ for $Y_i$ or not (and this decision $S_i$ could be based on the previous selection decisions $\F^{i-1}$ and on $I_i$). 
As long as for each $i$, we can construct marginally valid intervals in the sense of definition \eqref{def:marg-CI}, that is, 
\[
\PPst{Y_i \in \I_i(X_i,\alpha_i)}{\F^{i-1}} \leq \alpha_i,
\]
then the  FCR (or other variants) for the selected prediction intervals will be controlled in exactly the same fashion as we have proven for selected confidence intervals. This insight has particularly important ramifications for \emph{conformal prediction}, which is a way of constructing assumption-free marginal predictive intervals \citep{vovk2005algorithmic,shafer2008tutorial}. Indeed, it is also well known that assumption-free conditional inference is impossible \citep{vovk2012conditional,lei2014distribution,barber2019limits}. However, our constructions allow for assumption-free selective inference, which is a reasonable middle ground between fully marginal inference (standard conformal) and fully conditional inference (impossible). To avoid introducing an entirely new problem in the discussion, we include a few more details for the interested reader in Section~\ref{sec:conformal} of the Appendix.

\bibliography{FDR}
\bibliographystyle{agsm}

\appendix

\section{Extensions}

In this section, we discuss a few extensions that could be useful in practice. Specifically these are (a) post-hoc FCR guarantees, (b) decaying-memory FCR, (c) asynchronous FCR control, and (d) conformal prediction. The details are fairly straightforward and hence are left to the reader.

\subsection{Post-hoc FCR control}

We briefly consider the scenario when the coverage levels $\alpha_i$ were set in an arbitrary predictable fashion (possibly based on an online FCR algorithm, but possibly not). It is immediate from the results of \cite{katsevich2018towards} that one can construct post-hoc upper bounds on the achieved FCP that are uniformly valid over time with high probability. 

\begin{theorem}
Let $\alpha_1,\alpha_2,\dots$ be any sequence of predictable coverage levels, and $S_i$ be any sequence of selection decisions. Then for any constant $a > 0$ and confidence level $\delta \in (0,1)$, we have
\[
\PP{\forall n \geq 1, \FCP(n) \leq \frac{a+\sum_{i=1}^n \alpha_i}{\sum_{i=1}^n S_i} \cdot \frac{\log(1/\delta)}{a \log(1+\tfrac{\log(1/\delta)}{a})}  } \geq 1-\delta.
\]
\end{theorem}

We note that while $\FCP(n)$ is unknown, the above upper bound can be easily calculated and tracked by the user. Due to the logarithmic scaling in $\delta$ and small{} explicit constants (for example using $a=1$), these bounds are not just theoretical constructs but are also practically meaningful.

\subsection{Decaying-memory FCR}

In an online setup such as the one considered in this paper, where the horizon is potentially infinite, one may posit that recent selections and miscoverage events should matter much more to the data scientist than older ones. The current definition of FCR treats all selections and miscoverage events equally, and hence may not be ideal. A natural extension of the proposal by \cite{RYWJ17} is to consider an alternate error metric called the decaying-memory FCR (memFCR) at time $T$, that is defined for any decay parameter $\delta \leq 1$ as
\[
\textnormal{memFCR}(T) = \EE{\frac{\sum_{i=1}^T \delta^{T-i} V_i }{\sum_{i=1}^T \delta^{T-i} S_i}}.
\]
Naturally, when $\delta=1$, we recover the FCR as a special case. Simple modifications of the LORD-CI algorithm considered in this paper, exactly along the lines of those proposed in \cite{RYWJ17}, immediately result in online memFCR control.

\subsection{Asynchronous FCR control, and handling local dependence}\label{sec:asynchronous}
The setting considered in this paper can be thought of as the ``fully synchronized'' setting, meaning that when the $i$-th selection decision is being made, one already has knowledge of all previous decisions. However, as extensively discussed in \cite{zrnic2018asynchronous}, large scale testing is often asynchronous and parallel. 

Indeed, to estimate $\theta_i$, one may often collect the data sequentially (eg: clinical trial, A/B test) which can be time-consuming. While this data-collection is in progress, one may begin to run other parallel experiments to estimate $\theta_{i+1}, \theta_{i+2}, \dots$. When experiment $i+1$ starts, the coverage level $\alpha_{i+1}$ needs to be specified at the start, because that often determines when to stop, and the selection decision $S_{i+1}$ is made at the end. In this setup, when a new experiment is started, one may only have selection decisions corresponding an arbitrary subset of the experiments that have begun earlier. Further, we note that the data of experiments running in parallel may have an arbitrary dependence, for example if the two experiments happen to share subjects.

Using the \emph{principle of pessimism} described in \cite{zrnic2018asynchronous}, one can derive variants of the LORD-CI algorithm that provably control the mFCR in an arbitrary asynchronous environment, even when the data from one experiment depends arbitrarily on data from temporally nearby experiments.

\section{Selective conformal inference}\label{sec:conformal}

Conformal prediction is a general nonparametric technique for producing marginally valid prediction intervals under almost no regularity assumptions on the data generating process beyond exchangeability of the data. A simple version of the setup can be explained as follows. Let $(X_1,Y_1),\dots,(X_n,Y_n)$ be drawn i.i.d. from some joint distribution $P_{XY} = P_X \times P_{Y|X}$, which are supported on the domain $\mathcal{X} \times \mathcal{Y}$, where for simplicity let $\mathcal{Y}=\R$. Given a test point $X_{n+1}$ drawn i.i.d. from $P_X$, our task is to provide a prediction interval for the unobserved $Y_{n+1}$.

Conformal prediction begins by hallucinating a value $y$, to form a new dataset $(X_1,Y_1),\dots,(X_n,Y_n),(X_{n+1},y)$. One may then train any regression algorithm $f:\mathcal{X} \to \mathcal{Y}$ on this set of $n+1$ points to obtain $\widehat f$, and calculate the  $n+1$ in-sample residuals $r_i = Y_i - \widehat{f}(X_i)$ for $i \in [n]$ and $r_{n+1} = y - \widehat{f}(X_{n+1})$. We then ``reject'' $y$ if $r_{n+1}$ is in the largest $\alpha$-quantile of all $n+1$ residuals. We then repeat this whole process for every possible $y \in \mathcal{Y}$. The final prediction interval $\I(X_{n+1},\alpha)$ consists of all those $y$s that we did not reject. The intuition is that when $y=Y_{n+1}$, then all $n+1$ residuals are exchangeable, and so the rank of $r_{n+1}$ among $r_1,\dots,r_{n+1}$ is uniform. Hence the probability of rejecting $y=Y_{n+1}$ and excluding it from the interval, equals the probability that $r_{n+1}$ is in the largest $\alpha$-quantile of $r_1,\dots,r_{n+1}$, which is at most $\alpha$. The formal guarantee is that $I_{n+1} := \mathcal{I}(X_{n+1},\alpha)$ is marginally valid:
\[
\PP{Y_{n+1} \notin \I(X_{n+1},\alpha)} \leq \alpha,
\]
where the probability is taken over all $(n+1)$ draws from $P_{XY}$. Remarkably, this guarantee holds with no assumptions on the distribution $P_{XY}$ or on the regression algorithm $f$ (these may affect the length of the intervals, but not their validity). However, a conditional conformal guarantee is in general impossible, meaning that if we do not make any distributional assumptions and we would like a guarantee of the form
\[
\PP{Y_{n+1} \notin \I(X_{n+1},\alpha) | X_{n+1}=x} \leq \alpha
\]
to hold for any $x$, then the corresponding conditional conformal interval $\I(X_{n+1},\alpha)$ must have infinite length. The impossibility of fully conditional conformal prediction was pointed out by \cite{vovk2012conditional}, elaborated further by \cite{lei2014distribution}, and recently explored by \cite{barber2019limits}.

The relationship of the above discussion to the current paper is as follows. There was nothing in particular that restricted the setup of the current paper to confidence intervals for parameters $\theta_i$ based on observations $X_i$. The setup just as easily covers prediction intervals for outcomes $Y_i$ based on features $X_i$. To understand the implications, suppose we were to observe a sequence $X_{n+1},\dots,X_{n+m},\dots$ of test points drawn i.i.d. from $P_X$, and we do not wish to cover all of the corresponding $Y$s but just some subset of them. Then, one may construct marginally valid prediction intervals (at predictable levels $\alpha_i$ using LORD-CI) for an adaptively selected subset of $X_i$s, and this paper provides an FCR control guarantee on those selected intervals. 

Hence, even though conditional conformal inference is impossible, our works implies that ``selective conformal inference'' is possible. There is no contradiction here: an FCR guarantee is weaker than a conditional guarantee. Also, the FCR guarantee cannot really be used to get a conditional guarantee; indeed, if one was to only select $X_i$ for coverage if it is in a very small $\epsilon$-ball around a given point $x$ (to approximate the conditional coverage guarantee), then such selections would be very infrequent, and the $\alpha_i$ used would be very close to zero, resulting in an exceedingly wide interval. As $\epsilon\to 0$, we would also see $\alpha_i \to 0$, and thus the length of the selected interval would become infinite.

We end this section by remarking that nothing was particular to conformal prediction intervals; the FCR guarantees would also apply to any other marginally-valid prediction intervals. 
Further, there was also nothing particular to the online setting; such FCR control can also be guaranteed in the offline setting, but the latter analysis is out of the scope of this paper.

\end{document}